\documentclass[a4paper, 10pt, twosides]{amsart}

\usepackage{lmodern}
\usepackage{amssymb}
\usepackage{amsthm}

\usepackage{cite}
\usepackage{txfonts}
\usepackage{enumitem}
\usepackage{dsfont}
\usepackage{graphicx}
\usepackage[usenames, dvipsnames]{color}
\usepackage{tikz}
\usetikzlibrary{automata, positioning, shapes}
\usetikzlibrary{shapes,snakes,calc}
\usetikzlibrary{decorations.pathreplacing}
\usepackage{mathtools}
\allowdisplaybreaks
\usepackage{todonotes}
\usepackage{algorithm,algorithmic}

\newtheorem{theorem}{Theorem}
\newtheorem{proposition}{Proposition}

\newtheorem{defn}{Definition}
\newtheorem{example}{Example}

\newtheorem{prob}{Problem}

\newtheorem{rem}{Remark}

\newtheorem{assump}{Assumption}
\newtheorem{fact}{Fact}

\newtheorem{obsvn}{Observation}

\newcommand{\abs}[1]{\left\lvert{#1}\right\rvert}

\newcommand{\norm}[1]{\left\lVert#1\right\rVert}
\newcommand{\pmat}[1]{\begin{pmatrix}#1\end{pmatrix}}

\newcommand{\R}{\mathbb{R}}
\newcommand{\N}{\mathbb{N}}
\newcommand{\Svec}{\mathcal{S}}

\newcommand{\K}{\mathcal{K}}
\newcommand{\KL}{\mathcal{KL}}
\newcommand{\cL}{\mathcal{L}}

\newcommand{\G}{\mathcal{G}}
\newcommand{\V}{\mathcal{V}}
\newcommand{\E}{\mathcal{E}}

\newcommand{\lra}{\longrightarrow}

\newcommand{\ls}{\lambda_{i_{s}}}

\newcommand{\lu}{\lambda_{i_{u}}}

\newcommand{\is}{i_{s}}
\newcommand{\iu}{i_{u}}

\renewcommand{\l}{\ell}

\newcommand{\wv}{\overline{w}}
\newcommand{\we}{\underline{w}}

\newcommand{\Nsig}{N^{\gamma}_{t}}
\newcommand{\si}{\sigma_{i}}
\newcommand{\xp}{x_{i}}

\DeclareMathOperator{\minimize}{minimize}
\DeclareMathOperator{\sbjto}{subject\;to}

\DeclareMathOperator{\average}{average}

\allowdisplaybreaks
\title[]{Stabilizing Scheduling Policies for Networked Control Systems}

\author[A.\ Kundu]{Atreyee Kundu}
\author[D.\ E.\ Quevedo]{Daniel E. Quevedo}
\thanks{Atreyee Kundu is with the Department of Electrical Engineering, Indian Institute of Science Bangalore, India. Email: atreyeek@iisc.ac.in. Daniel E.\ Quevedo is with the Department of Electrical Engineering, Paderborn University, Germany. Email: dquevedo@ieee.org.}
\date{\today}

\begin{document}

    \begin{abstract}
        %\todo[inline]{Reviewer 1, comment 1:
        %Efficient allocation, but no analysis of efficiency%}
        This paper deals with {the problem of allocating} communication resources for Networked Control Systems (NCSs). We consider an NCS consisting of a set of discrete-time {LTI} plants whose stabilizing feedback loops are closed through a shared communication channel. Due to a limited communication capacity of the channel, not all plants can exchange information with their controllers at any instant of time. We propose a {method to find} periodic scheduling polic{ies} under which global asymptotic stability of each plant in the NCS is preserved. The individual plants are represented as switched systems, and the NCS is expressed as a weighted directed graph. We construct stabilizing scheduling polic{ies} by employing cycles on the underlying weighted directed graph of the NCS that satisfy {appropriate contractivity} conditions. We also discuss algorithmic design of these cycles.
    \end{abstract}
\maketitle

%================================================================
    \section{Introduction}
\label{s:intro}
    Networked Control Systems (NCSs) are omnipresent in modern day Cyber-Physical Systems (CPS) and Internet of Things (IoT) applications. While these applications typically involve a large-scale setting, the network resources are often limited. Consequently, multiple plants may need to share a communication channel {(or network)} for exchanging information with their remotely located controllers. {Examples of communication networks with limited bandwidth include wireless networks (an important component of smart home, smart transportation, smart city, remote surgery, platoons of autonomous vehicles, etc.) and underwater acoustic communication systems.} The often encountered scenario wherein the number of plants sharing a communication channel is higher than the capacity of the channel is called \emph{medium access constraint}. %\todo[inline]{Reviewer 1, comment 2; Reviewer 7, comment 1: Practical applications} %Stabilization of NCSs under medium access constraints has attracted considerable research attention in the past two decades.%, see e.g.\cite{Rehbinder'04,Goodwin'04,Lin'05,Cervin'06,Gorges'09,Al-Areqi'15, Hritsu'01, Zhang'06,Gaid'06,Dai'10} and the references therein.

    In this paper we consider an NCS consisting of multiple discrete-time linear plants whose feedback loops are closed through a shared communication channel. {A block diagram of such an NCS is shown in Figure \ref{fig:ncs}.}
    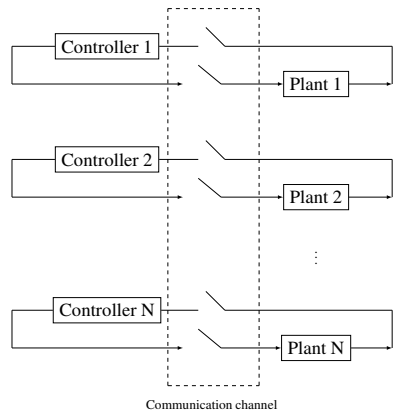
\begin{figure}[htbp]
    	\begin{center}
	\scalebox{0.5}{
	\begin{tikzpicture}[every path/.style={>=latex},base node/.style={draw,rectangle, scale = 1.4}]
	\node[base node] (a) at (-2,5) {Controller 1};
	\node[base node] (b) at (3.5,4) {Plant 1};
	\node[base node] (c) at (-2,2) {Controller 2};
	\node[base node] (d) at (3.5,1) {Plant 2};
	\node[base node] (e) at (-2,-2) {Controller N};
	\node[base node] (f) at (3.5,-3) {Plant N};	
	
	\draw (-4.5 ,5) edge (a);
	\draw (-4.5,5) edge (-4.5,4);
	\draw[->] (-4.5,4) -- (0,4);
	\draw (a) edge (0.4,5);
	\draw[->] (b) -- (5.5,4);
	\draw (5.5,4) edge (5.5,5);
	\draw[->] (1,4) -- (b);
	\draw[-.] (1,4) -- (0.4,4.5);
	\draw (5.5,5) edge (1.1,5);
	\draw[-.] (1.1,5) -- (0.6,5.5);

	\draw (-4.5 ,2) edge (c);
	\draw (-4.5,2) edge (-4.5,1);
	\draw[->] (-4.5,1) -- (0,1);
	\draw (c) edge (0.4,2);
	\draw[->] (d) -- (5.5,1);
	\draw (5.5,1) edge (5.5,2);
	\draw[->] (1,1) -- (d);
	\draw[-.] (1,1) -- (0.4,1.5);
	\draw (5.5,2) edge (1.1,2);
	\draw[-.] (1.1,2) -- (0.6,2.5);

	\draw (-4.5 ,-2) edge (e);
	\draw (-4.5,-2) edge (-4.5,-3);
	\draw[->] (-4.5,-3) -- (0,-3);
	\draw (e) edge (0.4,-2);
	\draw[->] (f) -- (5.5,-3);
	\draw (5.5,-3) edge (5.5,-2);
	\draw[->] (1,-3) -- (f);
	\draw[-.] (1,-3) -- (0.4,-2.5);
	\draw (5.5,-2) edge (1.1,-2);
	\draw[-.] (1.1,-2) -- (0.6,-1.5);

	\draw[dashed] (-0.4,-4) -- (-0.4,6);
	\draw[dashed] (2,-4) -- (2,6);
	\draw[dashed] (-0.4,-4) -- (2,-4);
	\draw[dashed] (-0.4,6) -- (2,6);

	\node (g) at (0.75,-4.5) {Communication channel};
	\node (h) at (3.5,-0.5) {\(\vdots\)};

	\end{tikzpicture}
	}
	\caption{Block diagram of NCS}\label{fig:ncs}
	\end{center}
    \end{figure}
 We assume that the plants are unstable in open-loop and asymptotically stable in closed-loop. Due to a limited communication capacity of the channel, only a few plants can exchange information with their controllers at any instant of time. Consequently, the remaining plants operate in open-loop at every time instant. Our objective is to allocate the shared communication channel to {the set of} plant{s} in a manner so that stability of {all} plant{s} is preserved. This task of efficient allocation of communication resources is commonly referred to as a \emph{scheduling problem}, {and the corresponding allocation scheme is called a \emph{scheduling policy}}.

 	{Scheduling policies that preserve qualitative behaviour of an NCS under limited communication and/or computation resources are widely researched upon, and tools from both control theory and communication theory have been explored, see the recent works \cite{Wen'16,Al-Areqi'15,Saha'15,He'16,Miao'17} and the references therein. These policies can be broadly classified into two categories: \emph{static} (also called \emph{periodic}, \emph{fixed}, or \emph{open-loop}) and \emph{dynamic} (also called \emph{non-periodic}, or \emph{closed-loop}) scheduling. In case of the former, a finite length allocation scheme of the shared communication channel is determined offline and is applied eternally in a periodic manner, while in case of the latter, the allocation is determined based on some information (e.g., states, outputs, access status of sensors and actuators, etc.) about the plant. In this paper we will focus on periodic scheduling policies that preserve global asymptotic stability (GAS) of all plants in an NCS. We will call such scheduling policies as \emph{stabilizing scheduling policies}. Static scheduling policies are easier to implement, often near optimal, and guarantee activation of each sensor and actuator, see \cite{Peters'16,Longo, Hristu'05} for detailed discussions. They are preferred for safety-critical control systems \cite[\S2.5.1]{Longo}. It is also observed in \cite{Orihuela'14, Peters'16} that periodic phenomenon appears in non-periodic schedules.}

 	%{Scheduling policies that preserve qualitative behaviour of an NCS are widely researched upon both under limited communication and computation resources, see \cite{Wen'16,Al-Areqi'15,Saha'15,Khatib'17} for recent works.}

 	%{Stabilizing scheduling policies are widely studied in the NCSs literature both under limited communication and computation resources, see \cite{Wen'16,Al-Areqi'15,Saha'15,Khatib'17} for some recent works. These policies can be broadly classified into two categories: \emph{static} (also called \emph{periodic}, \emph{fixed}, or \emph{open-loop}) and \emph{dynamic} (also called \emph{non-periodic} or \emph{closed-loop}) scheduling. In case of the former, a periodic allocation of the shared communication channel is determined offline and is applied eternally, while in case of the latter, the allocation is determined based on some information (e.g., states, outputs, access status of sensors and actuators, etc.) about the plant. In this paper we aim for stabilizing periodic scheduling policies. Static scheduling policies are easier to implement, often near optimal, and guarantee activation of each sensor and actuator, see \cite{Peters'16, Longo'13, Hristu'05} for detailed discussions. They are particularly preferred for safety-critical control systems \cite{Longo'13}.}
	
	{For NCSs with continuous-time linear plants, stabilizing periodic scheduling policies are characterized using common Lyapunov functions \cite{Hristu'01} and piecewise Lyapunov-like functions with average dwell time switching \cite{Lin'05}. A more general case of co-designing a static scheduling policy and control action is addressed using combinatorial optimization with periodic control theory in \cite{Rehbinder'04} and Linear Matrix Inequalities (LMIs) optimization with average dwell time technique in \cite{Dai'10}. In the discrete-time setting, the authors of \cite{Zhang'06} characterize periodic switching sequences that ensure reachability and observability of the plants under limited communication, and design an observer-based feedback controller for these periodic sequences. The techniques were later extended to the case of constant transmission delays \cite{Hristu'08} and Linear Quadratic Gaussian (LQG) control problem \cite{Hristu_Zhang'08}. Periodic sensor scheduling schemes that accommodate limited communication and adversary attacks are studied recently in \cite{Sid'17}.}
	
	%{Prior results on characterization of stabilizing periodic scheduling policies rely on common Lyapunov functions \cite{Hristu'01}, Rate Monotonic (RM) approach \cite{Branicky'02} and piecewise Lyapunov-like functions with average dwell time switching \cite{Lin'05} for the continuous-time setting. A more general case of co-designing a static scheduling policy and control action is addressed using combinatorial optimization with periodic control theory in \cite{Rehbinder'04} and Linear Matrix Inequalities (LMI) optimization with average dwell time technique in \cite{Dai'10}. This co-design problem is studied for NCSs with discrete-time linear plants in \cite{Zhang'06}. In the discrete-time setting, the authors of \cite{Zhang'06} characterize periodic switching sequences that ensure reachability and observability of the plants under communication constraints, and design an observer-based feedback controller for these periodic sequences. The techniques were later extended to the case of constant transmission delays \cite{Hristu'08} and Linear Quadratic Gaussian (LQG) control problem \cite{Hristu'08}.}
	
	The main contribution of this paper lies in combining switched systems and graph theory to propose a new class of stabilizing scheduling policies for NCSs. We represent the individual (open-loop unstable) plants of an NCS as switched systems, where the switching is between their open-loop (unstable mode) and closed-loop (stable mode) operations. Clearly, {within our setting,} no switched system can operate in stable mode for all time as that will destabilize some of the plants in the NCS. The search for a stabilizing scheduling policy then becomes the problem of finding switching logics that obey the limitations of the shared channel and preserve stability.  {It is assumed that the exchange of information between a plant and its controller is not affected by communication uncertainties.} In the recent past, graph theoretic techniques have played an important role in designing stabilizing switching logics for switched systems, see e.g., \cite{abc,def} and the references therein. {In the present work we} associate a weighted directed graph with the NCS that captures the communication limitation of the shared channel, and design stabilizing switching logic{s} for each plant in the NCS. Multiple Lyapunov-like functions are employed for analyzing stability of the switched systems. The stabilizing switching logics form a stabilizing scheduling policy. The switching logics are combined in terms of a class of cycles on the underlying weighted directed graph of the NCS that satisfies appropriate contractivity properties. We also discuss algorithmic construction of these cycles.

In brief, our contributions are:
    \begin{itemize}[label=\(\circ\),leftmargin=*]
    	\item Given an NCS with discrete-time linear plants that exchange information with their stabilizing controllers through a shared channel of limited communication capacity, we present an algorithm to design a scheduling policy that preserves GAS of each plant in the NCS. Our scheduling policy is periodic in nature, and relies on the existence of what we call a \(T\)-contractive cycle on the underlying weighted directed graph of the NCS. {Periodic scheduling policies are proven to be immensely useful in process control, where many loops need to share a common communication resource as it avoids the necessity of frequent network reconfigurations. In fact, periodic scheduling is an inherent feature of IEEE 802.15.4 networks \cite{Peters'16} which underlie commercial standards, such as WirelessHart, ISA100.11a and ZigBee.} {The use of cycles on a weighted directed graph makes our techniques numerically tractable, see Remark \ref{rem:compa} for a detailed discussion. }
    %\todo[inline]{Reviewer 1, comment 2; Reviewer 7, comment 1: Practical applications}
	\item We address algorithmic design of \(T\)-contractive cycles. Given the connectivity of the underlying weighted directed graph of the NCS and description of the individual plants, we fix a cycle on this graph and present an algorithm that designs multiple Lyapunov-like functions such that the above cycle is \(T\)-contractive. {We also identify sufficient conditions on the multiple Lyapunov-like functions and channel constraints under which the existence of a \(T\)-contractive cycle is guaranteed.}
    \end{itemize}
    %\todo[inline]{We added new results, hence, this new sentence in blue.}
    %Our stabilizing scheduling policy is \emph{static} in the sense that a \(T\)-contractive cycle on the underlying weighted directed graph of the NCS is computed off-line, and the scheduling policy is implemented by following certain logic involving this cycle.
    %To the best of our knowledge, this is the first instance in the literature where stabilization of NCSs is studied at the intersection of switched systems and graph theory. {Our algorithm is also useful in}

    The remainder of this paper is organized as follows: in \S\ref{s:prob_stat} we formulate the problem under consideration, and describe the primary apparatus for our analysis. Our method for constructing stabilizing periodic scheduling policies appears in \S\ref{s:stab_sched}. In \S\ref{s:T_contrac-design} we discuss algorithmic design of \(T\)-contractive cycles. Numerical examples are presented in \S\ref{s:num_ex} to demonstrate our results. We conclude in \S\ref{s:concln} with a brief discussion of future research directions. A proof of our main Theorem appears in \S\ref{s:all_proofs}.

    Some notation used in this paper: \(\N = \{1,2,\ldots\}\) is the set of natural numbers, \(\N_{0} = \{0\}\cup\N\), and \(\R\) is the set of real numbers. We let \(]k_{1}:k_{2}]\) denote the set \(\{n\in\N\:|\:k_{1}<n\leq k_{2}\}\). For a scalar \(m\), let \(\abs{m}\) denote its absolute value, and for a set \(M\), let \(\abs{M}\) denote its cardinality. Let \(\norm{\cdot}\) be the standard \(2\)-norm and \(^\top\) denote the transpose operation.

    \section{Preliminaries}
\label{s:prob_stat}
	We consider an NCS with \(N\) discrete-time linear plants. Each plant communicates with its remotely located controller through a shared communication channel, {see Figure \ref{fig:ncs}}.
    The plant dynamics are
    \begin{align}
    \label{e:sys_dyn}
    	x_{i}(t+1) = A_{i}x_{i}(t) + B_{i}u_{i}(t),\:x_{i}(0) = x_{i}^{0},\:i=1,2,\ldots,N,\:t\in\N_{0},
    \end{align}
     where \(x_{i}(t)\in\R^{d}\) and \(u_{i}(t)\in\R^{m}\) are the vectors of states and inputs of the \(i\)-th plant at time \(t\), respectively. Each plant \(i\) employs a state-feedback controller given by \(u_{i}(t) = K_{i}x_{i}(t)\), \(i=1,2,\ldots,N\). The matrices \(A_{i}\in\R^{d\times d}\), \(B_{i}\in\R^{d\times m}\) and \(K_{i}\in\R^{m\times d}\), \(i=1,2,\ldots,N\) are known.
     The shared channel has a limited communication capacity: at any time instant, only \(M\) plants (\(0 < M < N\)) can access the channel. Consequently, \(N-M\) plants operate in open loop at every time instant. We define
     \[
     	\Svec := \{s\in\{1,2,\ldots,N\}^{M}\:|\:\:\text{all elements of}\:\:s\:\:\text{are distinct}\}
    \]
    to be the set of vectors that consist of \(M\) distinct elements from \(\{1,2,\ldots,N\}\). We call a function \(\gamma:\N_{0}\to\Svec\) a {scheduling policy}. {There exists a diverging sequence of times \(0 =: \tau_{0}<\tau_{1}<\tau_{2}<\cdots\) and a sequence of indices \(s_{0},s_{1},s_{2},\ldots\) with \(s_{j}\in\Svec\), \(j = 0,1,2,\ldots\) such that \(\gamma(t) = s_{j}\) for \(t\in[\tau_{j}:\tau_{j+1}[\), \(j=0,1,2,\ldots\). } The role of \(\gamma\) is to specify, at every time {\(t\)}, \(M\) plants of the NCS which access the communication channel at that time. {The remaining \(N-M\) plants operate in open loop, in particular, with \(u_{i}(t) = 0\).}
    \begin{rem}
    \label{rem:diff_i/p_possibilities}
    \rm{
        One may also study a scheduling problem in the setting of the remaining \(N-M\) plants operating with \(u_{i}(t) = \overline{u}\), where \(\overline{u}\) is the last control input received before time \(t\). However, in this paper we consider open-loop evolution of a plant whenever it is not accessing the shared communication channel.
    }
    \end{rem}
   We will work under the following set of assumptions:
    \begin{assump}
     \label{a:stab-unstab}
     \rm{
     	The open-loop dynamics of each plant is unstable and each controller is stabilizing. More specifically, the matrices \(A_{i}+B_{i}K_{i}\), \(i=1,2,\ldots,N\) are Schur stable and the matrices \(A_{i}\), \(i=1,2,\ldots,N\) are unstable.\footnote{We call a matrix unstable{,} if it is not Schur stable.}
	}
     \end{assump}
     \begin{assump}
     \label{a:ideal-ch}
     \rm{
     	The shared communication channel is ideal in the sense that exchange of information between plants and their controllers is not affected by communication uncertainties.
	}
     \end{assump}
     In view of Assumption \ref{a:stab-unstab}, each plant in \eqref{e:sys_dyn} operates in two modes: stable mode when the plant has access to the shared communication channel and unstable mode when the plant does not have access to the channel. Let us denote the stable and unstable modes of the \(i\)-th plant as \(\is\) and \(\iu\), respectively, \(A_{\is} = A_{i}+B_{i}K_{i}\) and \(A_{\iu} = A_{i}\), \(i=1,2,\ldots,N\). %From the theory of switched systems, it is well-known that switching between these two modes arbitrarily, does not necessarily preserve stability of the \(i\)-th plant \cite[p.\ 19]{Liberzon}. Consequently, a scheduling policy plays an important role in ensuring stability of individual plants of the NCS under consideration.
     In this paper we are interested in a scheduling policy that guarantees {GAS} of each plant in \eqref{e:sys_dyn}. In particular, we study the following problem:
     \begin{prob}
    \label{prob:main}
    \rm{
        Given the matrices \(A_{i}\), \(B_{i}\), \(K_{i}\), \(i=1,2,\ldots,N\), and a number \(M (< N)\), find a scheduling policy that ensures global asymptotic stability (GAS) of each plant \(i\) in \eqref{e:sys_dyn}.
    }
    \end{prob}
	We will call a scheduling policy \(\gamma\) that is a solution to Problem \ref{prob:main}, as a stabilizing scheduling policy. Recall that
    \begin{defn}[{\cite[Lemma 4.4]{Khalil}}]
    \label{d:gas}
    \rm{
        The \(i\)-th plant in \eqref{e:sys_dyn} is GAS for a given scheduling policy \(\gamma\), if there exists a class \(\mathcal{KL}\) function \(\beta_{i}\) such that the following inequality holds:
        \begin{align}
        \label{e:gas}
            \norm{x_{i}(t)}\leq\beta_{i}(\norm{x_{i}(0)},t)\:\:\text{for all}\:x_{i}(0)\in\R^{d}\:\:\text{and}\:\:t\geq 0.\footnotemark
        \end{align}
    }
    \end{defn}
    \footnotetext{Recall {classes of functions} \cite[Chapter 4]{Khalil}: $\K := \{\phi:[0,+\infty[\to[0,+\infty[\big|\phi\:\text{is continuous, strictly increasing,}\:\phi(0) = 0\}, \mathcal{L} := \bigl\{\psi:[0,+\infty[\lra[0,+\infty[\:\:\big|\:\:\psi\:\:\text{is continuous and}\:\:\psi(s)\searrow 0\:\:\text{as}\:\:s\nearrow +\infty\bigr\}, \KL := \bigl\{\chi:[0,+\infty[^{2}\lra[0,+\infty[\:\:\big|\:\:\chi(\cdot,s)\in\K\:\:\text{for each}\:\:s\:\:\text{and}\:\:\chi(r,\cdot)\in\cL\:\:\text{for each}\:\:r\bigr\}$.}
    %\vspace*{-0.8cm}
    Towards solving Problem \ref{prob:main}, we express individual plants in \eqref{e:sys_dyn} as switched systems and associate a weighted directed graph with the NCS under consideration. Our solution to Problem \ref{prob:main} involves two steps:
    \begin{itemize}[label = \(\circ\), leftmargin = *]
        \item first, we present an algorithm that constructs a scheduling policy by employing what we call a \(T\)-contractive cycle on the underlying weighted directed graph of the NCS, and
        \item second, we show that a scheduling policy obtained from our algorithm ensures GAS of each plant in \eqref{e:sys_dyn}.
    \end{itemize}
    We also discuss algorithmic design of \(T\)-contractive cycles. In the remainder of this section we catalog our analysis tools.
%==========================
\subsection{Individual plants and switched systems}
\label{s:i-plant_and_sw-sys}
    The dynamics of the \(i\)-th plant in \eqref{e:sys_dyn} can be expressed as a switched system \cite[\S 1.1.2]{Liberzon}
	\begin{align}
	\label{e:i-plant_sw-sys}
		x_{i}(t+1) = A_{\sigma_{i}(t)}x_{i}(t),\:\:x_{i}(0) = x_{i}^{0},\:\sigma_{i}(t)\in\{\is,\iu\},\:t\in\N_{0},
	\end{align}
	where the subsystems are \(\{A_{\is},A_{\iu}\}\) and a switching logic \(\sigma_{i}:\N_{0}\to\{\is,\iu\}\) satisfies:
	\[
		\sigma_{i}(t) =
		\begin{cases}
			\is,\:\:&\text{if}\:\:i\:\:\text{is an element of}\:\:\gamma(t),\\
			\iu,\:\:&\text{otherwise}.
		\end{cases}
		%\:\:i = 1,2,\ldots,N.
	\]
Clearly, a switching logic \(\sigma_{i}\), \(i=1,2,\ldots,N\) is a function of the scheduling policy \(\gamma\). In order to ensure GAS of the individual plants, it therefore, suffices to look for a \(\gamma\) that renders each \(\sigma_{i}\) stabilizing in the following sense: \(\sigma_{i}\) guarantees GAS of switched system \eqref{e:i-plant_sw-sys} for each \(i=1,2,\ldots,N\). We recall the following facts from {recent} literature:
	\begin{fact}{\cite[Fact 1]{abc}}
	\label{fact:key}
		\rm{
		For each \(i=1,2,\ldots,N\), there exist pairs \((P_p,\lambda_{p})\), \(p\in\{\is,\iu\}\), where \(P_{p}\in\R^{d\times d}\) are symmetric and positive definite matrices, and \(0<\ls<1\), \(\lu \geq 1\), such that with
	\begin{align}
	\label{e:Lyap-func_defn}
		\R^{d}\ni\xi\longmapsto V_{p}(\xi) := \langle P_{p}\xi,\xi\rangle\in[0,+\infty[,
	\end{align}	
	we have
	\begin{align}
	\label{e:key_ineq1}
		V_{p}(z_{p}(t+1))\leq\lambda_{p}V_{p}(z_{p}(t)),\:\:t\in\N_{0},
	\end{align}
	and \(z_{p}(\cdot)\) solves the \(p\)-th recursion in \eqref{e:i-plant_sw-sys}, \(p\in\{\is,\iu\}\).
	}
	\end{fact}
	\begin{fact}{\cite[Fact 2]{abc}}
	\label{fact:compa}
	\rm{
		For each \(i=1,2,\ldots,N\), there exist \(\mu_{pq} \geq 1\) such that
		\begin{align}
		\label{e:key_ineq2}
			V_{q}(\xi)\leq\mu_{pq}V_{p}(\xi)\:\:\text{for all}\:\:\xi\in\R^{d}\:\:\text{and}\:\:p,q\in\{\is,\iu\}.
		\end{align}
	}
	\end{fact}
The functions \(V_{p}\), \(p\in\{\is,\iu\}\), \(i=1,2,\ldots,N\) are called Lyapunov-like functions, and they are widely used in stability theory of switched and hybrid systems \cite{Branicky98,Liberzon}. We will use properties of these functions described in Facts \ref{fact:key} and \ref{fact:compa}, in our analysis towards deriving a stabilizing scheduling policy. The scalars \(\lambda_{p}\), \(p\in\{\is,\iu\}\) give quantitative measures of (in)stability associated to (un)stable modes of operation of the \(i\)-th plant. Linear comparability of \(V_{p}\)'s in \eqref{e:key_ineq2} follows from the definition of \(V_{p}\), \(p\in\{\is,\iu\}\) in \eqref{e:Lyap-func_defn}. In \cite[Proposition 1]{abc} a \emph{tight} estimate of the scalars \(\mu_{pq}\), \(p,q\in\{\is,\iu\}\) was proposed to be \(\lambda_{\max}(P_{q}P_{p}^{-1})\), where \(\lambda_{\max}(M)\) denotes the maximum eigenvalue of a matrix \(M\in\R^{d\times d}\).
%========================
\subsection{NCS and directed graphs}
\label{s:ncs_digraph}
	Recall that a directed graph is a set of vertices connected by edges, where each edge has a direction associated to it. We connect a directed graph \(\G(\V,\E)\) with the NCS under consideration. \(\G(\V,\E)\) contains:
	\begin{itemize}[label = \(\circ\), leftmargin = *]
		\item {a} vertex set \(\V\) consisting of \(N\choose M\) vertices that are labelled distinctly. The label associated to a vertex \(v\) is given by \(L(v) = \{\l_v(1),\l_v(2),\ldots,\l_v(N)\}\), where \(\l_{v}(i) = \is\) for any \(M\) elements and \(\l_{v}(i) = \iu\) for the remaining \(N-M\) elements. Two labels \(L(u)\) and \(L(v)\) are equal if \(\l_{u}(i) = \l_{v}(i)\) for all \(i=1,2,\ldots,N\). By {the term} distinct labelling, we mean that \(L(u) = L(v)\) whenever \(u=v\in\V\).
 		\item {a}n edge set \(\E\) consisting of a directed edge \((u,v)\) from every vertex \(u\in\V\) to every vertex \(v\in\V\), \(v\neq u\).
	\end{itemize}
	The label \(L(v)\) corresponding to a vertex \(v\in\V\) gives a combination of \(M\) plants operating in stable mode and the remaining \(N-M\) plants operating in unstable mode. Since \(\V\) contains \(N\choose M\) vertices and the label associated to each vertex is distinct, it follows that the set of vertex labels consists of all possible combinations of \(M\) plants accessing the communication channel and \(N-M\) plants operating in open-loop. A directed edge \((u,v)\) from a vertex \(u\) to a vertex \(v\:(\neq u)\) corresponds to {a transition} from a set of \(M\) plants accessing the communication channel (as specified by \(L(u)\)) to another set of \(M\) plants accessing the communication channel (as specified by \(L(v)\)). In the sequel we may abbreviate \(\G(\V,\E)\) as \(\G\) if there is no risk of confusion.
	
We use functions \(\wv:\V\to\R^{N}\) and \(\we:\E\to\R^{N}\) to associate weights to the vertices and edges of \(\G\), respectively. They are defined as
	\begin{align}
	\label{e:vertex_wts}
		\wv(v) &= \pmat{\wv_{1}(v)\\\wv_{2}(v)\\\vdots\\\wv_{N}(v)},\:\:v\in\V,\:\:\text{where}\nonumber\\
		\wv_{i}(v) &=
			\begin{cases}
				-\abs{\ln\lambda_{\is}},\:\:\:\:&\text{if}\:\:\l_{v}(i) = \is,\\
				\abs{\ln\lambda_{\iu}},\:\:\:\:&\text{if}\:\:\l_{v}(i) = \iu,
			\end{cases}
			\quad i=1,2,\ldots,N,
			\intertext{and}
		\label{e:edge_wts}
		\we(u,v) &= \pmat{\we_{1}(u,v)\\\we_{2}(u,v)\\\vdots\\\we_{N}(u,v)},\:\:(u,v)\in\E,\:\:\text{where}\nonumber\\
		\we_{i}(u,v) &=
			\begin{cases}
				\ln\mu_{\is\iu},&\text{if}\:\:\l_{u}(i) = \is\:\text{and}\:\l_{v}(i) = \iu,\\
				\ln\mu_{\iu\is},&\text{if}\:\:\l_{u}(i) = \iu\:\text{and}\:\l_{v}(i) = \is,\\
				0,\:\:\:\:&\text{otherwise},
			\end{cases}
			\:\:i=1,2,\ldots,N.
	\end{align}
	Here \(\ls\), \(\lu\), and \(\mu_{\is\iu}\), \(\mu_{\iu\is}\), \(i=1,2,\ldots,N\) are as described in Facts \ref{fact:key} and \ref{fact:compa}, respectively.
    %\todo[inline]{Reviewer 6, comment 2): Justification for the choice of vertex and edge weights}
    \begin{rem}
    \label{rem:wt_choice}
    \rm{
        {We will aim for achieving GAS of each switched system \eqref{e:i-plant_sw-sys}, \(i=1,2,\ldots,N\). For this purpose, we will compensate the increase in \(V_{p}\), \(p\in\{\is,\iu\}\) caused by activation of unstable mode \(\iu\) and switches between stable and unstable modes (\(\is\) to \(\iu\) and \(\iu\) to \(\is\)) by the decrease in \(V_{p}\), \(p\in\{\is,\iu\}\) achieved by using the stable modes \(\is\), \(i=1,2,\ldots,N\). As a natural choice, the vertex (subsystem) weights of \(\G\) relate to the rate of increase/decrease of the Lyapunov-like functions \(V_{p}\) captured by the scalars \(\lambda_{p}\), \(p\in\{\is,\iu\}\), and the edge (switch) weights of \(\G\) relate to the ``jump'' between Lyapunov-like functions \(V_{p}\) and \(V_{q}\), \(p,q\in\{\is,\iu\}\) captured by the scalars \(\mu_{pq}\), \(p,q\in\{\is,\iu\}\), \(i=1,2,\ldots,N\).\footnote{The use of absolute value and natural logarithm is explained in context, see Remark \ref{rem:wt_choice2}.}}
    }
    \end{rem}
    %\todo[inline]{Reviewer 1, comment 4: differentiating between vertices on a graph and vertices that appear on a cycle}
    \begin{example}
    \label{ex:digraph_ex}
        \rm{
        Consider an NCS with \(N = 3\) and \(M = 2\). The corresponding directed graph \(\G\) has
        \begin{itemize}[label = \(\circ\), leftmargin = *]
            \item \({3\choose 2} = 3\) vertices, {\(\V = \{\overline{v}_{1},\overline{v}_{2},\overline{v}_{3}\}\)} with labels {\(L(\overline{v}_{1}) = \{1_{s},2_{s},3_{u}\}\), \(L(\overline{v}_{2}) = \{1_{s},2_{u},3_{s}\}\), \(L(\overline{v}_{3}) = \{1_{u},2_{s},3_{s}\}\)}, and
            \item \(6\) directed edges, {\(\E = \{(\overline{v}_{1},\overline{v}_{2}),(\overline{v}_{1},\overline{v}_{3}),(\overline{v}_{2},\overline{v}_{1}),
                (\overline{v}_{2},\overline{v}_{3})\),\((\overline{v}_{3},\overline{v}_{1}),(\overline{v}_{3},\overline{v}_{2})\}\)}.
        \end{itemize}
        A pictorial representation of \(\G\) is shown below.
    \begin{center}
            \scalebox{0.7}{
        \begin{tikzpicture}[every path/.style={>=latex},base node/.style={draw,circle}]
            \node[base node]            (a) at (-1.5,0)  { \(\overline{v}_1\) };
            \node[base node]            (b) at (1.5,0)  { \(\overline{v}_2\) };
            \node[base node]            (c) at (0,-2) { \(\overline{v}_3\) };

	        \draw[->] (a) edge (b);
            \draw[->] (b) edge[bend right] (a);
	        \draw[->] (a) edge (c);
            \draw[->] (c) edge[bend left] (a);
            \draw[->] (b) edge (c);
            \draw[->] (c) edge[bend right] (b);
        \end{tikzpicture}}
    \end{center}
    }
      \end{example}
   The weights associated to the vertices and edges of \(\G\) are:
   \begin{align*}
        \wv(\overline{v}_{1}) = \pmat{-\abs{\ln\lambda_{1_{s}}}\\-\abs{\ln\lambda_{2_{s}}}\\\abs{\ln\lambda_{3_{u}}}},\:\:&\:\: \wv(\overline{v}_{2}) = \pmat{-\abs{\ln\lambda_{1_{s}}}\\\abs{\ln\lambda_{2_{u}}}\\-\abs{\ln\lambda_{3_{s}}}},\:\:
        \wv(\overline{v}_{3}) = \pmat{\abs{\ln\lambda_{1_{u}}}\\-\abs{\ln\lambda_{2_{s}}}\\-\abs{\ln\lambda_{3_{s}}}},\\\intertext{and} \we(\overline{v}_{1},\overline{v}_{2}) = \pmat{0\\\ln\mu_{2_{s}2_{u}}\\\ln\mu_{3_{u}3_{s}}},\:\:&\:\:
        \we(\overline{v}_{1},\overline{v}_{3}) = \pmat{\ln\mu_{1_{s}1_{u}}\\0\\\ln\mu_{3_{u}3_{s}}},\we(\overline{v}_{2},\overline{v}_{1}) = \pmat{0\\\ln\mu_{2_{u}2_{s}}\\\ln\mu_{3_{s}3_{u}}},\\
        \we(\overline{v}_{2},\overline{v}_{3}) = \pmat{\ln\mu_{1_{s}1_{u}}\\\ln\mu_{2_{u}2_{s}}\\0},\:\:&\:\:\we(\overline{v}_{3},\overline{v}_{1}) = \pmat{\ln\mu_{1_{u}1_{s}}\\0\\\ln\mu_{3_{s}3_{u}}},\we(\overline{v}_{3},\overline{v}_{2}) = \pmat{\ln\mu_{1_{u}1_{s}}\\\ln\mu_{2_{s}2_{u}}\\0}.
   \end{align*}
	\begin{rem}
	\label{rem:dir-vs-undir_graph}
    \rm{
		By construction of \(\G\), it contains two directed edges \((u,v)\) and \((v,u)\) between every two vertices \(u,v\in\V\). Employing an undirected graph instead of a directed one may appear to be a natural choice here. However, the use of directed edges allows us to distinguish easily between the transitions \(\is\) to \(\iu\) and \(\iu\) to \(\is\), \(i=1,2,\ldots,N\), and assign weights to the corresponding edges accordingly. Notice that since the vertex labels are distinct, for every two vertices \(u,v\in\V\), there {exists} at least one \(i\) for which \(\we_{i}(u,v)\) and \(\we_{i}(v,u)\) are different, \(i\in\{1,2,\ldots,N\}\).
    }
	\end{rem}
    Recall that \cite[p.\ 4]{Bollobas}  a \emph{walk} on a directed graph \(G(V,E)\) is an alternating sequence of vertices and edges \(W = \tilde{v}_{0},\tilde{e}_{1}, \tilde{v}_{1}, \tilde{e}_{2}\),\(\tilde{v}_{2}\),\\\(\ldots, \tilde{v}_{\ell-1},\tilde{e}_{\ell},\tilde{v}_{\ell}\), where \(\tilde{v}_{m}\in V\), \(\tilde{e}_{m} = (\tilde{v}_{m-1},\tilde{v}_{m})\in E\), \(0 < m \leq \ell\). The \emph{length} of a walk is its number of edges, counting repetitions, e.g., the length of \(W\) is \(\ell\). The \emph{initial vertex} of \(W\) is \(\tilde{v}_{0}\) and the \emph{final vertex} of \(W\) is \(\tilde{v}_{\ell}\). If \(\tilde{v}_{\ell} = \tilde{v}_{0}\), we say that the walk is \emph{closed}. A closed walk is called a \emph{cycle} if the vertices \(\tilde{v}_{k}\), \(0 < k < n\) are distinct from each other and \(\tilde{v}_{0}\). We will use the following class of cycles on \(\G\) for constructing a stabilizing scheduling policy:
    %\todo[inline]{Reviewer 1, comment 5: range of \(n\) is added to the definition}
    \begin{defn}
    \label{d:good_walk}
    \rm{
    	A cycle \(W = v_{0}, (v_{0},v_{1}), v_{1},\cdots, v_{n-1}, (v_{n-1},v_{0})\),\( v_{0}\) on \(\G(\V,\E)\) is called \emph{\(T\)-contractive} if there exist {integers} \(T_{v_{j}} > 0\), \(j = 0,1,\ldots,n-1\), {\(2\leq n\leq \abs{\V}\)} such that the following set of inequalities is satisfied:
	\begin{align}
	\label{e:contra_defn}
		\Xi_{i}(W) :=  \sum_{j=0}^{n-1}\wv_{i}(v_{j})T_{v_j} +  \sum_{\substack{j=0\\v_{n} := v_{0}}}^{n-1}\we_{i}(v_{j},v_{j+1})< 0%\:\:\text{for all}\:\:
	\end{align}
	for all \(i = 1,2,\ldots,N\), where \(n\) is the length of \(W\), \(\wv(v_{j})\) is the weight associated to vertex \(v_{j}\), \(\wv_{i}(v_{j})\) is the \(i\)-th element of \(\wv(v_{j})\), and \(\we(v_{j},v_{j+1})\) is the weight associated to edge \((v_{j},v_{j+1})\), \(\we_{i}(v_{j},v_{j+1})\) is the \(i\)-th element of \(\we(v_{j},v_{j+1})\), \(i = 1,2,\ldots,N\), \(j = 0,1,\ldots,n-1\). We call the scalar \(T_{v_{j}}\) as the \(T\)-factor of vertex \(v_{j}\), \(j = 0,1,\ldots,n-1\).
    }
    \end{defn}
    We will employ the integers \(T_{v_{j}}\), \(j = 0,1,\ldots,n-1\) to associate a time duration with every vertex \(v_{j}\), \(j=0,1,\ldots,n-1\) that appear{s} in \(W\). This time duration will determine how long a set of \(M\) plants can access the shared communication channel while preserving {GAS} of all plants in the NCS under consideration. In the {present} discrete-time setting{,} the association of integers with time durations is natural.
    \begin{rem}
    \label{rem:T-factor_usage}
    \rm{
        Definition \ref{d:good_walk} is an extension of \cite[Definition 2]{def} to a set of \(N\) switched systems in the discrete-time setting. In \cite{def} the notion of a contractive cycle with \(T\)-factors chosen from a given interval of real numbers was used to study input/output-to-state stability (IOSS) of continuous-time switched nonlinear systems under dwell time restrictions.
        }
    \end{rem}
    % \todo[inline]{Comment: \texttt{``Quickly explain what Tvj will be used for. Do the Tvj have to be integers?}\\Response: A remark is added in this regard. Indeed, Tvj's are integers, it is now specified.}
   % {
    %\begin{rem}
    %\label{rem:T-factor_usage}
%        In this paper we will employ the integers \(T_{v_{j}}\), \(j = 0,1,\ldots,n-1\) to associate a time duration with every vertex \(v_{j}\), \(j=0,1,\ldots,n-1\). This time duration will determine how long a set of \(M\) plants can access the shared communication channel while preserving global asymptotic stability of all plants in the NCS under consideration. In the discrete-time setting the association of integers with time durations is natural.
%        }
    %\end{rem}
    \begin{example}
    \label{ex:T-contrac}
    \rm{
        Consider the setting of Example \ref{ex:digraph_ex}. Suppose that
        \begin{align*}
            \lambda_{1_{s}} = \lambda_{2_{s}} = \lambda_{3_{s}} = 0.25,\:\:&\:\:\lambda_{1_{u}} = \lambda_{2_{u}} = \lambda_{3_{u}} = 1.1,\\
            \mu_{1_{s}1_{u}} = \mu_{2_{s}2_{u}} = \mu_{3_{s}3_{u}} = 1.1,\:\:&\:\:\mu_{1_{u}1_{s}} = \mu_{2_{u}2_{s}} = \mu_{3_{u}3_{s}} = 1.2.
        \end{align*}
        The cycle \(W = 1,(1,2),2,(2,1),1\) on \(\G\) is \(T\)-contractive with \(T\)-factors \(T_{1} = 5\) and \(T_{2} = 4\). Indeed, \(\Xi_{1}(W) = -1.3863\), \(\Xi_{2}(W) = -6.2726\) and \(\Xi_{3}(W) = -4.791\).
        }
    \end{example}
   % }

   We are now in a position to describe our solution to {Problem~\ref{prob:main}}.

    %\section{Main results}
%\label{s:mainres}
\section{Stabilizing {periodic} scheduling polic{ies}}
\label{s:stab_sched}
    %\todo[inline]{Comment: To add periodic in the section heading.\\Response: Done.}
	 The following algorithm is geared towards constructing a {periodic} scheduling policy. We will show that a scheduling policy obtained from this algorithm is stabilizing. {Let \(\G(\V,\E)\) be a directed graph representation of the {NCS} described in \S\ref{s:prob_stat}.Suppose that \(\G(\V,\E)\) admits a \(T\)-contractive cycle \(W = v_{0},(v_{0},v_{1}),v_{1},\ldots,v_{n-1},(v_{n-1},v_{0}),v_{0}\) (of length \(n\)) with \(T\)-factors \(T_{0}\), \(T_{1}, \ldots\), \(T_{n-1}\).}
%	\begin{algorithm}
%		\caption{Construction of stabilizing scheduling policy}
%		\begin{algorithmic}
%		\renewcommand{\KwData}{\textbf{Input:}}
%		\renewcommand{\KwResult}{\textbf{Output:}}
%		\KwData{A contractive}\\
%		\KwResult{A periodic scheduling policy}
%		\end{algorithmic}
%	\end{algorithm}
	\begin{algorithm}[htbp]
    %\todo[inline]{Reviewer 1, comment 4: \(u_{j}\) is changed to \(s_{j}\)}
    %\todo[inline]{Reviewer 1, comment 6: \(T\)-factors are also inputs now}
			\caption{Construction of a periodic scheduling policy} \label{algo:sched_algo}
		\begin{algorithmic}[1]
			\renewcommand{\algorithmicrequire}{\textbf{Input:}}
			\renewcommand{\algorithmicensure}{\textbf{Output:}}
			
			\REQUIRE a \(T\)-contractive cycle {\(W=v_{0},(v_{0},v_{1}),v_{1},\ldots,v_{n-1}\),\((v_{n-1},v_{0}),v_{0}\) and corresponding \(T\)-factors \(T_{0}, T_{1},\ldots,T_{n-1}\)}
			\ENSURE a periodic scheduling policy \(\gamma\)
			
			 \hspace*{-0.6cm}\textit {Step I: For each vertex \(v_{j}\), \(j = 0,1,\ldots,n-1\), pick the elements
			 \(i\) with label \(\l_{v_{j}}(i) = \is\), \(i=1,2,\ldots,N\), and construct \(M\)-
			 dimensional vectors {\(s_{j}\)}, \(j = 0,1\),\(\ldots,n-1\).}
			 \FOR{ \(j = 0,1,\ldots,n-1\)}
			 	\STATE Set \(p = 0\)
				\FOR {\(i = 1,2,\ldots,N\)}
					\IF {\(\l_{v_{j}}(i) = \is\)}
					\STATE Set \(p=p+1\) and \({s_{j}(p)} = i\)
					\ENDIF
				\ENDFOR
			 \ENDFOR
			
			 \hspace*{-0.6cm}\textit{Step II: Construct a scheduling policy using the vectors {\(s_{j}\)},\(j = 0,1,\ldots,n-1\) obtained in Step I and the \(T\)-factors \(T_{v_{j}}\),\(j = 0,1,\ldots,n-1\)}
			 \STATE Set \(p=0\) and \(\tau_{0} = 0\)
				 \FOR {\(q = pn, pn+1,\ldots, (p+1)n-1\)}			\label{algo_step:rec}
			 	\STATE Set \(\gamma(\tau_{q}) = {s_{q-pn}}\) and \(\tau_{q+1} = \tau_{q} + T_{v_{q-pn}}\)
				\STATE Output \(\tau_{q}\) and \(\gamma(\tau_{q})\)
			 \ENDFOR
			 \STATE Set \(p = p+1\) and go to \ref{algo_step:rec}.
		\end{algorithmic}
	\end{algorithm}

    Given a set of matrices \(A_{i}\), \(B_{i}\), \(K_{i}\), \(i = 1,2,\ldots,N\) and a number \(M\), Algorithm \ref{algo:sched_algo} employs a \(T\)-contractive cycle \(W = v_{0},(v_{0},v_{1}),v_{1},\ldots,v_{n-1},(v_{n-1},v_{0}),v_{0}\) on \(\G(\V,\E)\) to construct a scheduling policy \(\gamma\) that specifies, at every time, \(M (< N)\) plants that access the shared communication channel. The construction of \(\gamma\) involves two steps: in {Step I}, corresponding to each vertex \(v_{j}\), \(j = 0,1,\ldots,n-1\), a vector \(s_{j}\), \(j = 0,1,\ldots,n-1\) is created. The vector \(s_{j}\) contains the elements \(i\in\{1,2,\ldots,N\}\) for which \(\l_{v_{j}}(i) = \is\), where \(\l_{v_{j}}(i)\) denotes the \(i\)-th element of the vertex label \(L(v_{j})\). Recall that by construction, each \(L(v_{j})\) contains \(\l_{v_{j}}(i) = \is\) exactly for \(M\) \(i\)'s. Consequently, the length of \(s_{j}\) is \(M\), \(j = 0,1,\ldots,n-1\). In {Step II}, a scheduling policy \(\gamma\) is obtained from the vectors \(s_{j}\), \(j = 0,1,\ldots,n-1\) and the \(T\)-factors \(T_{v_{j}}\), \(j = 0,1,\ldots,n-1\). Sets of \(M\) plants corresponding to the elements in \(s_{j}\) access the shared communication channel for \(T_{v_{j}}\) duration of time, \(j = 0,1,\ldots,n-1\). In particular, the following mechanism is employed to construct values of \(\gamma\) on the intervals \([\tau_{pn}:\tau_{(p+1)n}[\), \(p = 0,1,\ldots\):
	\begin{align*}
		\begin{rcases}
			\gamma(\tau_{q}) &= s_{q-pn}\\
			\tau_{q+1} &= \tau_{q} + T_{v_{q}-pn}
		\end{rcases}
		q = pn, pn+1,\ldots, (p+1)n-1.
	\end{align*}
	Clearly, a scheduling policy \(\gamma\) constructed as above, is periodic with period \(\displaystyle{\sum_{j=0}^{n-1}T_{v_{j}}}\). A pictorial representation of a scheduling policy obtained from Algorithm \ref{algo:sched_algo} is given in Figure \ref{fig:policy_exmple}.%, {where activation of \(s_{j}\) corresponds to activation of the plants whose indices appear in \(s_{j}\)}.
   % \todo[inline]{Reviewer 1, comment 7: Figure modified to cater to non-equal \(T\)-factors; interpretation of Figure to be explained.}
   \begin{figure}
    \begin{center}
    \scalebox{0.8}{
        \begin{tikzpicture}
           \draw[thick,->] (0,0) -- (8,0);
           \draw[thick,->] (0,0) -- (0,4);

          %\draw[very thin] (0,0.5) -- (8,0.5);
          %\draw[very thin] (0,1) -- (8,1);
          %\draw[very thin] (0,3) -- (8,3);

          \node (a) at (-0.3,0.5) {$s_{0}$};
          \node (b) at (-0.3,1) {$s_{1}$};
          \node (c) at (-0.3,3) {$s_{n-1}$};
          \node (d) at (0,4.2) {$\gamma(t)$};
          \node (f) at (-0.2,-0.1) {$0$};
          \node (g) at (-0.3,2) {$\vdots$};

          \node (e) at (8.2,0) {$t$};
          \draw[very thin] (0,-0.2) -- (0,0.1);
          \draw[very thin] (1.1,-0.2) -- (1.1,0.1);
          \draw[very thin] (2,-0.2) -- (2,0.1);
          \draw[very thin] (4,-0.2) -- (4,0.1);
          \draw[very thin] (5.5,-0.2) -- (5.5,0.1);
          \draw[very thin] (6.6,-0.2) -- (6.6,0.1);
          \draw[very thin] (7.5,-0.2) -- (7.5,0.1);

          \draw[thick] (0,0.5) -- (1.1,0.5);
          \draw[thick] (1.1,1) -- (2,1);
          \draw[thick] (4,3) -- (5.5,3);
          \draw[thick] (5.5,0.5) -- (6.6,0.5);
          \draw[thick] (6.6,1) -- (7.5,1);

          \node (e) at (3,2) {$\cdots$};

          \draw [gray,decorate,decoration={brace,amplitude=5pt},
            xshift=0pt,yshift=0pt] (1.1,-0.2)  -- (0,-0.2)
            node [black,midway,below=4pt,xshift=-2pt] {\footnotesize $T_{v_{0}}$};
          \draw [gray,decorate,decoration={brace,amplitude=5pt},
            xshift=0pt,yshift=0pt] (2,-0.2)  -- (1.1,-0.2)
            node [black,midway,below=4pt,xshift=-2pt] {\footnotesize $T_{v_{1}}$};
          \draw [gray,decorate,decoration={brace,amplitude=5pt},
            xshift=0pt,yshift=0pt] (5.5,-0.2)  -- (4,-0.2)
            node [black,midway,below=4pt,xshift=-2pt] {\footnotesize $T_{v_{n-1}}$};
          \draw [gray,decorate,decoration={brace,amplitude=5pt},
            xshift=0pt,yshift=0pt] (6.6,-0.2)  -- (5.5,-0.2)
            node [black,midway,below=4pt,xshift=-2pt] {\footnotesize $T_{v_{0}}$};
          \draw [gray,decorate,decoration={brace,amplitude=5pt},
            xshift=0pt,yshift=0pt] (7.5,-0.2)  -- (6.6,-0.2)
            node [black,midway,below=4pt,xshift=-2pt] {\footnotesize $T_{v_{1}}$};

          \node (f) at (7.7,2) {$\cdots$};
        \end{tikzpicture}}
    \end{center}
    \caption{An example scheduling policy: activation of \(s_{j}\) corresponds to activation of the plants whose indices appear in \(s_{j}\)}\label{fig:policy_exmple}
    \end{figure}
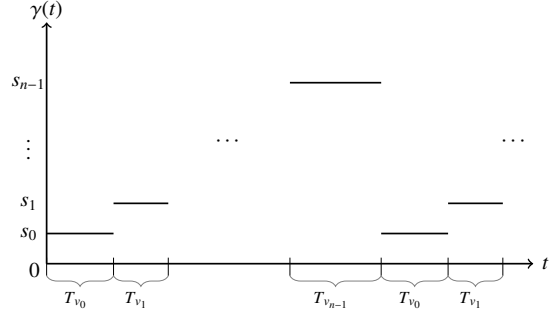
    %\todo[inline]{Extension to aperiodic cases}
      The following theorem asserts that a scheduling policy obtained from Algorithm \ref{algo:sched_algo} is a solution to Problem \ref{prob:main}.
	\begin{theorem}
	\label{t:mainres1}
	\rm{
		Consider a{n NCS} described in \S\ref{s:prob_stat}. Let the matrices \(A_{i}\), \(B_{i}\), \(K_{i}\), \(i = 1,2,\ldots,N\) and a number \(M (< N)\) be given. Then each plant in \eqref{e:sys_dyn} is {GAS} under a scheduling policy \(\gamma\) obtained from Algorithm \ref{algo:sched_algo}.
	}
	\end{theorem}
	A proof of Theorem \ref{t:mainres1} is provided in \S\ref{s:all_proofs}. For an NCS consisting of \(N\) discrete-time linear plants that are open loop unstable and closed-loop stable, and a shared communication channel that allows access only to \(M (< N)\) plants at every time instant, Algorithm \ref{algo:sched_algo} constructs a periodic scheduling policy that ensures GAS of each plant in the NCS.
	\begin{rem}
	\label{rem:static_policy}
	\rm{
		Our stabilizing scheduling policy is \emph{static} {and thereby easy to implement: A} \(T\)-contractive cycle on the underlying weighted directed graph of the NCS is computed off-line, and the scheduling policy is implemented by following logic{s} involving this cycle.
	}
	\end{rem}
	%\begin{rem}
	%\label{rem:compa}
	%\rm{
	%}
	%\end{rem}
    %\todo[inline]{To add: comparison with prior works}
     The existence of a stabilizing scheduling policy proposed in {this section}, depends on the existence of a \(T\)-contractive cycle on the underlying directed graph \(\G\) of the NCS. It is, therefore, of importance to study how to detect/design a \(T\)-contractive cycle on \(\G\). We address this matter next.
%=======================================
\section{Algorithmic design of \(T\)-contractive cycles}
\label{s:T_contrac-design}
    Given the weighted directed graph \(\G\), existence of a \(T\)-contractive cycle depends on two factors: connectivity of \(\G\) (for existence of cycles) and the weights associated to the vertices and edges of \(\G\) (for \(T\)-contractivity of cycles). Since \(\G\) is a complete graph by construction, it necessarily admits cycles. Fix a cycle \(W = v_{0}, (v_{0},v_{1}), v_{1}\),\(\ldots,v_{n-1},(v_{n-1}\),\\\(v_{0}),v_{0}\) on \(\G\). The \(T\)-contractivity of \(W\) is guaranteed by the existence of integers \(T_{v_{j}} > 0\), \(j = 0,1,\ldots,n-1\) such that condition \eqref{e:contra_defn} is satisfied. Existence of such \(T_{v_{j}}\)'s {depend upon} the vertex and edge weights \(\wv(v)\), \(v\in\V\) and \(\we(u,v)\), \((u,v)\in\E\) associated to \(\G\). These weights are functions of the matrices \(P_{p}\) and the scalars \(\lambda_{p}\), \(p\in\{\is,\iu\}\), \(i=1,2,\ldots,N\).
    \begin{rem}
	\label{rem:G_admits_T-contra-cycles}
    \rm{
		 Notice that the Lyapunov-like functions \(V_{p}\) and consequently, the scalars \(\lambda_{p}\), \(p\in\{\is,\iu\}\) and \(\mu_{pq}\), \(p,q\in\{\is,\iu\}\), \(i=1,2,\ldots,N\)  {used in \eqref{e:contra_defn}} are not unique. {For each \(i\in\{1,2,\ldots,N\}\), we} have that \(A_{\is}\) is Schur stable and \(A_{\iu}\) is unstable. It is known that a Schur stable matrix \(A\in\R^{d\times d}\) satisfies the following \cite[Proposition 11.10.5]{Bernstein}: for every symmetric and positive definite matrix \(Q\in\R^{d\times d}\), there exists a symmetric and positive definite matrix \(P\in\R^{d\times d}\) such that the discrete-time Lyapunov equation
    \begin{align}
    \label{e:disc_Lyap_eqn}
        A^\top P A - P + Q = 0
    \end{align}
holds. For a pre-selected symmetric and positive definite matrix \(Q_{\is}\), let \(P_{\is}\) be the solution to \eqref{e:disc_Lyap_eqn} with \(A = A_{\is}\), \(P = P_{\is}\) and \(Q = Q_{\is}\); we put \(V_{\is}(\xi) := \xi^\top P_{\is}\xi\) as the corresponding Lyapunov-like function. {Direct calculations} along with an application of the standard inequality \cite[Lemma 8.4.3]{Bernstein} leads to the estimate \(\displaystyle{\ls = 1 - \frac{\lambda_{\min}(Q_{\is})}{\lambda_{\max}(P_{\is})}}\), which satisfies \(0 < \ls < 1\). Similarly, for the unstable matrix \(A_{\iu}\), there exists \(0 < \eta < 1\) such that \(\eta A_{\iu}\) is Schur stable. Fix a symmetric and positive definite matrix \(Q_{\iu}\). Let \(P_{\iu}\) be the solution to \eqref{e:disc_Lyap_eqn} with \(A = \eta A_{\iu}\), \(P = P_{\iu}\) and \(Q = Q_{\iu}\); we put \(V_{\iu}(\xi) := \xi^\top P_{\iu}\xi\) as the corresponding Lyapunov-like function. A straightforward calculation gives an estimate \(\displaystyle{\lu = \frac{1}{\eta^{2}} > 1}\). Clearly, the choice of the matrices \(Q_{p}\), \(p\in\{\is,\iu\}\) determines the choice of the matrices \(P_{p}\), \(p\in\{\is,\iu\}\) and the scalars \(\lambda_{p}\), \(p\in\{\is,\iu\}\), \(i = 1,2,\ldots,N\). In addition, the matrices \(P_{p}\), \(p\in\{\is,\iu\}\) determine the scalars \(\mu_{pq}\), \(p,q\in\{\is,\iu\}\) as described in \S\ref{s:i-plant_and_sw-sys}.{\qed} %Our aim is to ``co-design'' the matrices \(P_{p}\) and the scalars \(\lambda_{p}\), \(p\in\{\is,\iu\}\), \(i = 1,2,\ldots,N\) such that \(\G\) admits a \(T\)-contractive cycle.
    }
    \end{rem}
    Recall that \(\G\) has \(N\choose M\) vertices. Consequently, depending on the values of \(N\) and \(M\), one may need to design a \(T\)-contractive cycle on a ``large'' directed graph for implementing the scheduling policy proposed in {\S\ref{s:stab_sched}}. It is clear that checking for existence of \(T_{v_{j}}\), \(j=0,1,\ldots,n-1\) corresponding to all possible values of \(\lambda_{p}\), \(p\in\{\is,\iu\}\), \(\mu_{pq}\), \(p,q\in\{\is,\iu\}\), \(i = 1,2,\ldots,N\) for every cycle \(W\) on \(\G\), is not numerically tractable. {To overcome this issue, we will next} address the design of a \(T\)-contractive cycle on \(\G\) in two steps:
    \begin{itemize}[label = \(\circ\), leftmargin = *]
    	\item first, we identify conditions on the scalars \(\lambda_{p}\), \(p\in\{\is,\iu\}\) and \(\mu_{pq}\), \(p,q\in\{\is,\iu\}\), \(i=1,2,\ldots,N\) under which a cycle on \(\G\) satisfying certain properties, is \(T\)-contractive, and
	\item second, given the matrices \(A_{i}\), \(B_{i}\), \(K_{i}\), \(i=1,2,\ldots,N\), we present an algorithm to design the scalars \(\lambda_{p}\), \(p\in\{\is,\iu\}\) and \(\mu_{pq}\), \(p,q\in\{\is,\iu\}\), \(i=1,2,\ldots,N\) such that the above conditions are met.
    \end{itemize}
    %first, we identify sufficient conditions for \(T\)-contractivity of a cycle \(W\) on \(\G\), and second, we present an algorithm to design the vertex and edge weights \(\wv(v)\), \(v\in\V\) and \(\we(u,v)\in\E\) of \(\G\) such that the proposed sufficient conditions are met.
%\subsection{Sufficient conditions for \(T\)-contractivity}
%\label{ss:cand_contrac}
    %Recall that for a vertex \(v\in\V\), \(\l_{v}(m)\) denotes the \(m\)-th element of its label \(L(v)\). We define
    \begin{defn}
    \label{d:T-contrac_candidate}
    \rm{
    	A cycle \(W = v_{0},(v_{0},v_{1}),v_{1},\ldots,v_{n-1},(v_{n-1},v_{0})\),\(v_{0}\) on \(\G(\V,\E)\) is called \emph{candidate contractive}{,} if for each \(i = 1,2,\ldots,N\), there {exists} at least one \(v_{j}\), \(j\in\{ 0,1,\ldots,n-1\}\) such that \(\l_{v_{j}}(i) = \is\).}
    \end{defn}

    In view of Definition \ref{d:good_walk}, for \(T\)-contractivity of \(W = v_{0}, (v_{0},v_{1}), v_{1}, \ldots, (v_{n-1},v_{0}),v_{0}\), we require that the condition \(\Xi_{i}(W) < 0\) holds for all \(i=1,2,\ldots,N\). Since for each \(i=1,2,\ldots,N\), the scalars \(\ln\lambda_{\iu}\), \(\ln\mu_{\is\iu}\), \(\ln\mu_{\iu\is}\geq 0\), existence of at least one \(v_{j}\), \(j\in\{0,1,\ldots,n-1\}\) in \(W\) such that \(\ell_{v_{j}}(i) = \is\), is necessary. A candidate contractive cycle satisfies this property. %A candidate contractive cycle is \(T\)-contractive depending on the vertex and edge weights \(\wv(v)\), \(v\in\V\) and \(\we(u,v)\), \((u,v)\in\E\) of \(\G\).
    Fix an \(i\in\{1,2,\ldots,N\}\). %Let the matrices \(P_{p}\) and the scalars \(\lambda_{p}\), \(p\in\{\is,\iu\}\) be given, and the scalars \(\mu_{pq}\), \(p,q\in\{\is,\iu\}\) be computed as \(\mu_{pq} = \lambda_{\max}(P_{q}P_{p}^{-1})\).
    We let \(\overline{N}_{pq}\) denote the total number of times \(\l_{v_{j}}(i) = p\) and \(\l_{v_{j+1}}(i) = q\) appear in \(W\), \(p,q\in\{\is,\iu\}\), \(j = 0,1,\ldots,n-1\), \(v_{n} := v_{0}\). %Let \(T^{S}_{m}, T^{U}_{m} > 0\) be given integers. Suppose that there exist integers \(0 < \overline{T}_{\is} \leq T^{S}_{m}\) and \(\overline{T}_{\iu} \geq T^{U}_{m} > 0\) such that the following inequality holds:
    %Our next result provides sufficient conditions on the scalars \(\lambda_{p}\), \(p\in\{\is,\iu\}\) and \(\mu_{pq}\), \(p,q\in\{\is,\iu\}\), \(i=1,2,\ldots,N\) for \(T\)-contractivity of \(W\).
    \begin{obsvn}
    \label{prop:mainres2}
    \rm{
        %Consider the Networked Control System described in \S\ref{s:prob_stat}, and its underlying weighted directed graph \(\G(\V,\E)\).
        Let \(W = v_{0},(v_{0},v_{1}),v_{1},\ldots,v_{n-1},(v_{n-1},v_{0}),v_{0}\) be a candidate contractive cycle on \(\G\). Suppose that there exist integers \(T_{v_{j}} > 0\), \(j=0,1,\ldots,n-1\) such that the following set of inequalities holds:
        \begin{align}
        \label{e:main_ineq1}
            &-\abs{\ln\ls}\Biggl(\sum_{\substack{j=0,1,\ldots,n-1|\\\l_{v_{j}}(i) = \is}}T_{v_{j}}\Biggr) + \abs{\ln\lu}\Biggl(\sum_{\substack{j=0,1,\ldots,n-1|\\\l_{v_{j}}(i) = \iu}}T_{v_{j}}\Biggr)
            + (\ln\mu_{\is\iu})\overline{N}_{\is\iu} + (\ln\mu_{\iu\is})\overline{N}_{\iu\is} < 0,\:\:i=1,2,\ldots,N,
        \end{align}
        where the scalars \(\lambda_{p}\), \(p\in\{\is,\iu\}\) and \(\mu_{pq}\), \(p,q\in\{\is,\iu\}\), \(i=1,2,\ldots,N\) are as described in Facts \ref{fact:key} and \ref{fact:compa}, respectively. Then \(W\) is \(T\)-contractive with \(T\)-factors \(T_{v_{j}}\) associated to the vertices \(v_{j}\), \(j=0,1,\ldots,n-1\).{\qed}
    }
    \end{obsvn}
    %See \S\ref{s:all_proofs} for a proof of Proposition \ref{prop:mainres2}. Fix a candidate contractive cycle \(W = v_{0},(v_{0},v_{1}),v_{1},\ldots,v_{n-1},(v_{n-1},v_{0}),v_{0}\) on \(\G\). Proposition \ref{prop:mainres2} asserts that if one can find integers \(T_{v_{j}}\), \(j=0,1,\ldots,n-1\) such that the scalars \(\lambda_{p}\), \(p\in\{\is,\iu\}\) and \(\mu_{pq}\), \(p,q\in\{\is,\iu\}\), \(i=1,2,\ldots,N\) satisfy condition \eqref{e:main_ineq1}, then \(W\) is \(T\)-contractive.
    In view of the definitions of vertex and edge weights \(\wv(v)\), \(v\in\V\) and \(\we(u,v)\), \((u,v)\in\E\) of \(\G\), the above observation follows immediately from \eqref{e:contra_defn}. A stabilizing scheduling policy \(\gamma\) constructed by employing the cycle \(W\) is periodic with period \(\displaystyle{\sum_{j=0}^{n-1}T_{v_{j}}}\). Notice that we do not consider the terms \(\overline{N}_{pq}\), \(p,q\in\{\is,\iu\}\), \(p=q\) for the candidate contractive cycle \(W\), which is no loss of generality. Indeed, from \cite[Proposition 1]{abc}, we have that \(\ln\mu_{\is\is} = \ln\mu_{\iu\iu} = 0\), \(i=1,2,\ldots,N\). %As highlighted in Remark \ref{rem:G_admits_T-contra-cycles}, the choice of the matrices \(P_{p}\) and the scalars \(\lambda_{p}\), \(p\in\{\is,\iu\}\), (and hence the scalars \(\mu_{pq}\), \(p,q\in\{\is,\iu\}\)), \(i=1,2,\ldots,N\) is not unique. Consequently, given the matrices \(A_{i}\), \(B_{i}\), \(K_{i}\), \(i=1,2,\ldots,N\) and a candidate contractive cycle \(W = v_{0},(v_{0},v_{1}),v_{1},\ldots,v_{n-1},(v_{n-1},v_{0}),v_{0}\), one would like to design the matrices \(P_{p}\) and the scalars \(\lambda_{p}\), \(p\in\{\is,\iu\}\), \(i=1,2,\ldots,N\) such that condition \eqref{e:main_ineq1} holds. We address this design problem in the remainder of this section.

%\subsection{Design of \((P_{p},\lambda_{p})\), \(p\in\{\is,\iu\}\), \(i=1,2,\ldots,N\)}
%\label{ss:weight_design}
     Given the matrices \(A_{i}\), \(B_{i}\), \(K_{i}\), \(i=1,2,\ldots,N\), and a candidate contractive cycle \(W\), our next algorithm finds pairs \((P_{p},\lambda_{p})\), \(p\in\{\is,\iu\}\), \(i=1,2,\ldots,N\) such that condition \eqref{e:main_ineq1} holds. %The scalars \(\mu_{pq}\), \(p,q\in\{\is,\iu\}\), \(i=1,2,\ldots,N\) are computed by using the estimate proposed in \cite[Proposition 1]{abc}.
     \begin{rem}
     \label{rem:BMI_set}
     \rm{
        The pairs \((P_{p},\lambda_{p})\), \(p\in\{\is,\iu\}\) are solutions to the following set of Bilinear Matrix Inequalities (BMI):
        \begin{align}
        \label{e:BMIset}
            \begin{aligned}
                A_{\is}^\top P_{\is}A_{\is} - \ls P_{\is}&\preceq 0,\:\:P_{\is}\succ 0,\:\:0<\ls<1,\\
                A_{\iu}^\top P_{\iu}A_{\iu} - \lu P_{\iu}&\preceq 0,\:\:P_{\iu}\succ 0,\:\:\lu\geq 1.
            \end{aligned}
        \end{align}
    In general, solving BMIs is a numerically difficult task. We will use a grid-based approach, where the BMIs are transformed into LMIs --- solution tools for which are widely available.
    }
    \end{rem}
   % Consider the {NCS} described in \S\ref{s:prob_stat}, and its underlying weighted directed graph \(\G(\V,\E)\).
    \begin{algorithm}[htbp]
	\caption{Design of a \(T\)-contractive cycle} \label{algo:scalar_design_algo}
    		\begin{algorithmic}[1]
    			\renewcommand{\algorithmicrequire}{\textbf{Input}:}
			\renewcommand{\algorithmicensure}{\textbf{Output}:}
	
			\REQUIRE matrices \(A_{i}\), \(B_{i}\), \(K_{i}\), \(i=1,2,\ldots,N\), a candidate contractive cycle \(W = v_{0},(v_{0},v_{1}),v_{1},\ldots,v_{n-1},(v_{n-1},v_{0}),v_{0}\)
			\ENSURE \(T\)-factors for \(W\)
			
			\hspace*{-0.6cm}\textit{Step I: Compute the matrices \(A_{\is}\) and \(A_{\iu}\), \(i=1,2,\ldots,N\)}
			\FOR {\(i=1,2,\ldots,N\)}
				\STATE Set \(A_{\is} = A_{i} + B_{i}K_{i}\) and \(A_{\iu} = A_{i}\)
			\ENDFOR
			
			\hspace*{-0.6cm}\textit{Step II: Compute the integers \(\overline{N}_{pq}\), \(p,q\in\{\is,\iu\}\), \(i=1,2,\ldots,N\)}
			\FOR {\(i=1,2,\ldots,N\)}
				\STATE Compute \(\overline{N}_{pq}\), \(p,q\in\{\is,\iu\}\) from \(W\)
			\ENDFOR
			
			\hspace*{-0.6cm}\textit{Step III: Fix a set of values for \(\ls\in]0,1[\), \(i=1,2,\ldots,N\)}
				\STATE Fix a step-size \(h_{s} > 0\) (small enough) and compute \(k_{s} > 0\) such that \(k_{s}\) is the maximum integer satisfying \(k_{s}h_{s} < 1\)
				\FOR {\(i=1,2,\ldots,N\)}
					\STATE Set \(\Lambda_{i}^{S} = \{h_{s},2h_{s},\ldots,k_{s}h_{s}\}\)
				\ENDFOR
				
			\hspace*{-0.6cm}\textit{Step IV: Fix a set of values for \(\lu\in[1,+\infty[\), \(i=1,2,\ldots,N\)}
				\STATE Fix a step-size \(h_{u} > 0\) (small enough) and compute \(k_{u} > 0\) such that \(k_{u}\) is the maximum integer satisfying \(k_{u}h_{u} < 1\)
				\FOR {\(i=1,2,\ldots,N\)}
					\STATE Set \(\Lambda_{i}^{U} = \emptyset\)
					\FOR { \(\eta_{i} = h_{u},2h_{u},\ldots,k_{u}h_{u}\)}
						\IF {\(\eta_{i}A_{i}\) is Schur stable}
							\STATE Add element \(\frac{1}{\eta_{i}^{2}}\) to the set \(\Lambda_{i}^{U}\)
						\ENDIF
					\ENDFOR
				\ENDFOR
				
			\hspace*{-0.6cm}\textit{Step V: Check for pairs \((P_{p},\lambda_{p})\), \(p\in\{\is,\iu\}\), \(i=1,2,\ldots,N\) under which \(W\) is \(T\)-contractive}
				\FOR {all pairs \((\ls,\lu)\), \(\ls\in\Lambda_{i}^{S}\), \(\iu\in\Lambda_{i}^{U}\),\\\hspace*{0.2cm}\(i = 1,2,\ldots,N\)}
					\STATE Solve the following feasibility problem in \(P_{p}\), \(p\in\{\is,\iu\}\):
					 	\begin{align}
	           		 		\label{e:feasprob1}
		              				\minimize\:\:&\:\:1\nonumber\\
		              				\sbjto\:\:&\:\:
		                  			\begin{cases}
                              					A_{\is}^\top P_{\is}A_{\is} - \ls P_{\is}\preceq 0,\\
                              					A_{\iu}^\top P_{\iu}A_{\iu} - \lu P_{\iu}\preceq 0,\\
                              					P_{\is}, P_{\iu} \succ 0,\\
                              					\kappa I\preceq P_{\is},P_{\iu}\preceq I,\kappa > 0
		                  			\end{cases}
	             				\end{align}
					 \IF {there is a solution to \eqref{e:feasprob1}}
					 	\STATE Compute \(\mu_{\is\iu} = \lambda_{\max}(P_{\iu}P_{\is}^{-1})\) and \(\mu_{\iu\is} = \lambda_{\max}(P_{\is}P_{\iu}^{-1})\)
						\STATE Solve the following feasibility problem in \(T_{v_{j}}\), \(j=0,1,\ldots,n-1\):
						 	\begin{align}
	            					\label{e:feasprob2}
		              					\minimize\:\:&\:\:1\nonumber\\
		              					\sbjto\:\:&\:\:
		                  				\begin{cases}
                            						T_{v_{j}} > 0,\:\:j=0,1,\ldots,n-1,\\
                            						\text{condition}\:\eqref{e:main_ineq1}.\\
                            %(\ls,\lu,\mu_{\is\iu},\mu_{\iu\is})\in\mathcal{M}_{i},\\
                            %\:\:i=1,2,\ldots,N.
                          					\end{cases}
	             					\end{align}
						\IF {there is a solution to \eqref{e:feasprob2}}
					 		\STATE Output \(T_{v_{j}}\), \(j=0,1,\ldots,n-1\) and exit\\\hspace*{0.5cm}Algorithm \ref{algo:scalar_design_algo}
						\ENDIF
					\ENDIF
				\ENDFOR
    		\end{algorithmic}
   \end{algorithm}

    In Algorithm \ref{algo:scalar_design_algo} we employ a grid-based approach to design the pairs \((P_{\is},\ls)\) and \((P_{\iu},\lu)\) such that with the definition \eqref{e:Lyap-func_defn}, inequality \eqref{e:key_ineq1} holds.\footnote{Alternatively, one could also use the \emph{path-following method} proposed in \cite{Boyd_BMI}.}  The scalars \(\ls\) and \(\lu\) are varied over the sets \(\Lambda_{i}^{S}\) and \(\Lambda_{i}^{U}\), respectively. %Recall that for an unstable matrix \(A\in\R^{d\times d}\), there exist values \(\eta\in]0,1[\) such that \(\eta A\) is Schur stable. The above fact is used to determine
The elements of \(\Lambda_{i}^{S}\) belong to the interval \(]0,1[\), while the set \(\Lambda_{i}^{U}\) is determined as follows: a scalar \(\eta_{i}\) is varied over \(]0,1[\) with step-size \(h_{u}\), and the estimates \(\frac{1}{\eta_{i}^{2}}\) satisfying \(\eta_{i} A_{i}\) is Schur stable are stored in \(\Lambda_{i}^{U}\). For a fixed pair \((\ls,\lu)\) with \(\ls\in\Lambda_{i}^{S}\) and \(\lu\in\Lambda_{i}^{S}\), the following set of LMIs is solved:
    \begin{align}
    \label{e:LMI_set}
    	\begin{aligned}
    		A_{\is}^\top P_{\is}A_{\is} - \ls P_{\is} &\preceq 0,\\
		A_{\iu}^\top P_{\iu}A_{\iu} - \lu P_{\iu} &\preceq 0.
	\end{aligned}
    \end{align}
    If a solution to \eqref{e:LMI_set} {is found}, then the scalars \(\mu_{\is\iu}\) and \(\mu_{\iu\is}\) are computed using the estimates given in \cite[Proposition 1]{abc}. The feasibility problem \eqref{e:feasprob2} is then solved with the above estimates of \(\ls,\lu,\mu_{\is\iu},\mu_{\iu\is}\). If there is a solution to \eqref{e:feasprob2}, then the values of \(T_{v_{j}}\), \(j=0,1,\ldots,n-1\) are stored and Algorithm \ref{algo:scalar_design_algo} terminates. Otherwise, the pair \((\ls,\lu)\) is updated and the above process is repeated.
    \begin{rem}
        \label{rem:addnl_condn}
        \rm{
            The condition \(\kappa I\preceq P_{\is}, P_{\iu}\preceq I\) in the feasibility problem \eqref{e:feasprob1} is not inherent to the inequalities \eqref{e:BMIset}. It is included for numerical reasons, in particular, \(\kappa I\preceq P_{\is}, P_{\iu}\) limits the condition numbers of \(P_{\is}\) and \(P_{\iu}\) to {\({\kappa}^{-1}\)}, and the condition \(P_{\is}, P_{\iu}\preceq I\) guarantees that the set of feasible \(P_{\is},P_{\iu}\) is bounded.
            }
        \end{rem}
        \begin{rem}
        \label{rem:scalar_design_limit}
        \rm{
        		{Notice that even if the step-sizes \(h_{s}\) and \(h_{u}\) are chosen to be very small,} only a finite number of possibilities for \((P_{p},\lambda_{p})\), \(p\in\{\is,\iu\}\), \(i=1,2,\ldots,N\) are explored in Algorithm \ref{algo:scalar_design_algo}. Consequently, if no solution to { the feasibility problem \eqref{e:feasprob2} is found}, it is not immediate whether there are indeed no pairs \((P_{\is},\ls)\) and \((P_{\iu},\lu)\), \(i=1,2,\ldots,N\) for the given matrices \(A_{i}\), \(B_{i}\), \(K_{i}\) such that there are integers \(T_{v_{j}}\), \(j=0,1,\ldots,n-1\) satisfying condition \eqref{e:main_ineq1}. Algorithm \ref{algo:scalar_design_algo}, therefore, offers only a partial solution to the problem of designing suitable matrices \(P_{p}\) and the scalars \(\lambda_{p}\), \(p\in\{\is,\iu\}\), \(i=1,2,\ldots,N\) in the sense that the algorithm does not conclude about their non-existence. {It is, therefore, of interest to identify sufficient conditions under which the feasibility problem \eqref{e:feasprob2} admits a solution. We discuss this matter next.}
        }
        \end{rem}
%============================================================
%============================================================
%\subsubsection{Solution to \eqref{e:feasprob2}}
%\label{ss:feas_soln}
   % \todo[inline,inlinewidth = 9cm]{Reviewer 1, comment 8; Reviewer 6, comment 1; Reviewer 7, comment 3: (additional results) sufficient conditions for existence of \(T\)-contractive cycles}
        {Existence of a solution to the feasibility problem \eqref{e:feasprob2} depends on the choice of a candidate contractive cycle \(W\) and the scalars \(\lambda_{p}\), \(p\in\{\is,\iu\}\), \(i=1,2,\ldots,N\).\footnote{Notice that while the scalars \(\mu_{pq}\), \(p,q\in\{\is,\iu\}\), \(i=1,2,\ldots,N\) affect the choice of \(T\)-factors that solve the feasibility problem \eqref{e:feasprob2}, they do not affect the existence of a solution to \eqref{e:feasprob2}. Indeed, given the scalars \(\lambda_{p}\), \(p\in\{\is,\iu\}\) and \(\mu_{pq}\), \(p,q\in\{\is,\iu\}\), \(i=1,2,\ldots,N\), and a candidate contractive cycle \(W = v_{0},(v_{0},v_{1}),v_{1},\ldots,v_{n-1},(v_{n-1},v_{0}), v_{0}\) on \(\G\), if there exists \(T_{v_{j}} = \tilde{T}\), \(j=0,1,\ldots,n-1\) such that condition \eqref{e:main_ineq1} holds, then it follows that condition \eqref{e:main_ineq1} holds for any \(T'\geq\tilde{T}\).} The first component above is governed by the given numbers \(M\) and \(N\). Recall that for a vertex \(v\in\V\), \(\l_{v}(m)\) denotes the \(m\)-th element of its label \(L(v)\). Let \(v^{s}\) denote the set of elements \(j_{1},j_{2},\ldots,j_{M}\in\{1,2,\ldots,N\}\) satisfying \(\l_{v}(j_{p}) = j_{p_{s}}\), \(p=1,2,\ldots,N\). Below we propose a set of sufficient conditions on the scalars \(\lambda_{p}\), \(p\in\{\is,\iu\}\), \(i=1,2,\ldots,N\) and the number \(M\) under which the feasibility problem \eqref{e:feasprob2} admits a solution.
        \begin{proposition}
        \label{prop:mainres3}
        \rm{
            Let \(M = 1\). Consider a candidate contractive cycle \(W = v_{0},(v_{0},v_{1}),v_{1},\ldots,v_{N-1},(v_{N-1},v_{0}),v_{0}\) on \(\G(\V,\E)\) that satisfies \(v^{s}_{k}\cap v^{s}_{\l} = \emptyset\) for all \(k,\ell = 0,1,\ldots,N-1\), \(k\neq\l\). Suppose that the scalars \(\lambda_{p}\), \(p\in\{\is,\iu\}\), \(i=1,2,\ldots,N\) satisfy
            \begin{align}
            \label{e:suff_condn1}
                \abs{\ln\lambda_{\is}} - (N-1)\abs{\ln\lambda_{\iu}} > 0,\:\:i=1,2,\ldots,N.
            \end{align}
            Then there exists \(\tilde{T}\in\N\) such that the cycle \(W\) is \(T\)-contractive with \(T_{v_{j}} = \tilde{T} > 0\), \(j=0,1,\ldots,N-1\).
        }
        \end{proposition}
        %Proposition \ref{prop:mainres3} provides conditions on the scalars \(\lambda_{p}\), \(p\in\{\is,\iu\}\), \(i=1,2,\ldots,N\) and the number \(M\) under which \(\G\) admits a \(T\)-contractive cycle \(W = v_{0},(v_{0},v_{1}),v_{1},\ldots,v_{N-1},(v_{N-1},v_{0}), v_{0}\). For each \(i\in\{1,2,\ldots,N\}\), \(W\) contains exactly one vertex \(v_{j}\) with \(\l_{v_{j}}(i) = \is\), \(j=0,1,\ldots,N-1\). A short proof of Proposition \ref{prop:mainres3} is presented in \S\ref{s:all_proofs}.
        %
        %We further observe that if the number of plants accessing the shared communication channel is at least half of the total number of plants, then one can relax condition \eqref{e:suff_condn1} and design a \(T\)-contractive cycle \(W\) on \(\G\) of length \(2\).
        \begin{proposition}
        \label{prop:mainres4}
        \rm{
            Let \(\displaystyle{M\geq{N}/{2}}\). Consider a candidate contractive cycle \(W = v_{0},(v_{0},v_{1}),v_{1},(v_{1},v_{0}),v_{0}\) on \(\G(\V,\E)\) that satisfies \(v_{1}^{s}\supset\{1,2,\ldots,N\}\setminus v_{0}^{s}\). Suppose that the scalars \(\lambda_{p}\), \(p\in\{\is,\iu\}\), \(i=1,2,\ldots,N\) satisfy
            \begin{align}
            \label{e:suff_condn2}
                \abs{\ln\lambda_{\is}} - \abs{\ln\lambda_{\iu}} > 0,\:\:i=1,2,\ldots,N.
            \end{align}
            Then there exists \(\tilde{T}\in\N\) such that the cycle \(W\) is \(T\)-contractive with \(T_{v_{0}} = T_{v_{1}} = \tilde{T}\).
        }
        \end{proposition}
        %Proposition \ref{prop:mainres4} deals with \(T\)-contractivity of a cycle \(W\) on \(\G\) that has two vertices: \(v_{0}\) and \(v_{1}\). Let \(j_{1},j_{2},\ldots,j_{M}\) be the indices for which \(\l_{v_{0}}(j_{p}) = j_{p_{s}}\), \(p=1,2,\ldots,M\). Then \(v_{1}^{s}\) contains all elements of \(\{1,2,\ldots,N\}\setminus v_{0}^{s}\). Since \(\displaystyle{M\geq\frac{N}{2}}\), two vertices suffice to activate stable modes of operation of all plants \(i=1,2,\ldots,N\). The cycle \(W\) above is \(T\)-contractive if condition \eqref{e:suff_condn2} holds. Notice that condition \eqref{e:suff_condn2} is a relaxed version of condition \eqref{e:suff_condn1}. Instead of requiring \(\abs{\ln\lambda_{\is}}\) is strictly bigger than \((N-1)\) times \(\abs{\ln\lambda_{\iu}}\), we rely on \(\abs{\ln\lambda_{\is}}\) being strictly bigger than \(\abs{\ln\lambda_{\iu}}\). We present a short proof of Proposition \ref{prop:mainres4} in \S\ref{s:all_proofs}.
        Proposition \ref{prop:mainres3} deals with the case when exactly one plant is allowed to access the shared communication channel at any time instant, while Proposition \ref{prop:mainres4} deals with the case where at least half of the total number of plants have access to the shared communication channel. In case of the former, a \(T\)-contractive cycle contains exactly one vertex \(v_{j}\) with \(\l_{v_{j}}(i) = \is\) for each \(i\), \(j=0,1,\ldots,N-1\), while in case of the latter, \(\l_{v_{j}}(i) = \is\) for each \(i\), is accommodated in two vertices, \(j=0,1\), \(i=1,2,\ldots,N\). Condition \eqref{e:suff_condn2} is a relaxation of condition \eqref{e:suff_condn1}. We present concise proofs of Propositions \ref{prop:mainres3} and \ref{prop:mainres4} in \S\ref{s:all_proofs}.
        \begin{example}
        \label{ex:suff_condn}
        \rm{
            Consider \(N=3\) with
            \begin{align*}
                (A_{1},B_{1},K_{1}) &= \Biggl(\pmat{0.2 & 0.7\\1.6 & 0.1},\pmat{1\\0},\pmat{-0.2752 & -0.6705}\Biggr),\\
                (A_{2},B_{2},K_{2}) &= \Biggl(\pmat{1 & 0.1\\0.1 & 1},\pmat{0\\1},\pmat{-0.9137 & -0.9505}\Biggr),\\
                (A_{3},B_{3},K_{3}) &= \Biggl(\pmat{1.2 & 0.2\\0.1 & 0.9},\pmat{1\\0},\pmat{-1.0757 & -0.4839}\Biggr).
            \end{align*}
            Corresponding to \(V_{p}(\xi) = \xi^\top P_{p}\xi\), \(p\in\{\is,\iu\}\), \(i=1,2,3\), we obtain the following estimates of the scalars \(\lambda_{p}\), \(p\in\{\is,\iu\}\) and \(\mu_{pq}\), \(p,q\in\{\is,\iu\}\), \(i=1,2,3\):
            \begin{align*}
                \lambda_{1_{s}} = 0.2787,\:\:\lambda_{1_{u}} = 1.5625,\:\:
                \mu_{1_{s}1_{u}} = 4.1786,\:\:\mu_{1_{u}1_{s}} = 1.5338,\\
                \lambda_{2_{s}} = 0.0859,\:\:\lambda_{2_{u}} = 1.2346,\:\:
                \mu_{2_{s}2_{u}} = 23.5578,\:\:\mu_{2_{u}2_{s}} = 1.9130,\\
                \lambda_{3_{s}} = 0.2147,\:\:\lambda_{3_{u}} = 2.0408.\:\:
                \mu_{3_{s}3_{u}} = 3.6524,\:\:\mu_{3_{u}3_{s}} = 2.5238.
            \end{align*}

            Let \(M = 1\). We have that condition \eqref{e:suff_condn1} holds. Indeed,
            \begin{align*}
                \abs{\ln\lambda_{1_{s}}} - 2\abs{\ln\lambda_{1_{u}}} &= 0.3850 > 0,\\
                \abs{\ln\lambda_{2_{s}}} - 2\abs{\ln\lambda_{2_{u}}} &= 2.0331 > 0,\\
                \abs{\ln\lambda_{3_{s}}} - 2\abs{\ln\lambda_{3_{u}}} &= 0.1118 > 0.
            \end{align*}
            The cycle \(W_{1} = v_{0}, (v_{0},v_{1}), v_{1}, (v_{1},v_{2}), v_{2}, (v_{2},v_{0}), v_{0}\), where \(\l_{v_{0}}(1) = 1_{s}\), \(\l_{v_{1}}(2) = 2_{s}\) and \(\l_{v_{2}}(3) = 3_{s}\), is \(T\)-contractive with \(T_{v_{0}} = T_{v_{1}} = T_{v_{2}} = \tilde{T} = 20\). We have \(\Xi_{1}(W_{1}) = -6.0596\), \(\Xi_{2}(W_{1}) = -36.85\), \(\Xi_{3}(W_{1}) = -0.0154\).

            Now, let \(M = 2 (>\frac{N}{2})\). Since the scalars \(\lambda_{p}\), \(p\in\{\is,\iu\}\), \(i=1,2,3\) satisfy \eqref{e:suff_condn1}, it is immediate that \eqref{e:suff_condn2} holds. The cycle \(W_{2} = v_{0},(v_{0},v_{1}),v_{1},(v_{1},v_{0}),v_{0}\), where \(\l_{v_{0}}(1) = 1_{s}\), \(\l_{v_{0}}(2) = 2_{s}\) and \(\l_{v_{1}}(2) = 2_{s}\), \(\l_{v_{1}}(3) = 3_{s}\), is \(T\)-contractive with \(T_{v_{0}} = T_{v_{1}} = \tilde{T} = 5\). Indeed, \(\Xi_{1}(W_{2}) = -2.2990\), \(\Xi_{2}(W_{2}) = -24.5457\), \(\Xi_{3}(W_{2}) = -1.9047\).
            }
        \end{example}
        \begin{rem}
        \label{rem:T-factors}
        \rm{
            Both in Propositions \ref{prop:mainres3} and \ref{prop:mainres4}, we consider the simplest setting where the \(T\)-factors associated to all vertices that appear in \(W = v_{0},(v_{0},v_{1}),v_{1},\ldots,v_{n-1},(v_{n-1},v_{0}),v_{0}\), are the same. However, this choice of \(T\)-factors can also be extended to non-equal \(T_{v_{j}}\), \(j=0,1,\ldots,n-1\). For instance, in Example \ref{ex:suff_condn}, the candidate contractive cycle \(W_{2} = v_{0},(v_{0},v_{1}),v_{1},(v_{1},v_{0}),v_{0}\) is also \(T\)-contractive with \(T_{v_{0}} = 5\) and \(T_{v_{1}} = 4\). It follows that \(\Xi_{1}(W_{2}) = -2.7452\), \(\Xi_{2}(W_{2}) = -22.0911\), \(\Xi_{3}(W_{2}) = -0.3662\).
        }
        \end{rem}}
        %\begin{rem}
%        \label{rem:limitations}
%        \rm{
%        {
%            Our design of stabilizing scheduling policies relies on the existence of a \(T\)-contractive cycle \(W\) on \(\G\). Given the matrices \(A_{i}\), \(B_{i}\), \(K_{i}\), \(i=1,2,\ldots,N\) and the underlying weighted directed graph \(\G\) of an NCS, if Algorithm \ref{algo:scalar_design_algo} does not yield a solution for \emph{any} candidate contractive cycle \(W\) on \(\G\), then our technique for constructing a stabilizing periodic scheduling policy fails.}
%        }
%        \end{rem}
        \begin{rem}
        \label{rem:compa}
        \rm{
        		{Switched systems have appeared before in NCSs literature, see e.g., \cite{Ishii'02, Dai'10, Lin'05, Zhang'06}, and average dwell time switching logic is proven to be a useful tool. In the presence of unstable systems, stabilizing average dwell time switching involves two conditions on \(]0:t]\) for every \(t\in\N\) \cite{Liberzon_IOSS}: i) an upper bound on the number of switches and ii) a lower bound on the ratio of durations of activation of stable to unstable subsystems. In contrast, our design of a stabilizing scheduling policy involves design of a \(T\)-contractive cycle on the underlying weighted directed graph of the NCS. To design these cycles, we solve the feasibility problems \eqref{e:feasprob1} and \eqref{e:feasprob2}. Condition \eqref{e:contra_defn} does not involve nor imply restrictions on the behaviour of a scheduling policy on every interval \(]0:t]\), \(t\in\N\).}
        		%On the one hand, in \cite{Dai'09,Dai'10} a switched delay system model of NCS is employed to address medium access constraints and network induced uncertainties. The design of a stabilizing scheduling policy involves satisfaction of a set of matrix inequalities and average dwell time condition{s} for each plant of the NCS. Average dwell time switching in the presence of unstable subsystems involves two conditions on \(]0:t]\) for every \(t\in\N\): i) an upper bound on the number of switches and ii) a lower bound on the ratio of durations of activation of stable to unstable subsystems. On the other hand, in the absence of communication uncertainties, our design of a stabilizing scheduling policy involves design of a \(T\)-contractive cycle on the underlying weighted directed graph of the NCS. To design these cycles, we solve the feasibility problems \eqref{e:feasprob1} and \eqref{e:feasprob2}. Condition \eqref{e:contra_defn} does not involve nor imply restrictions on the behaviour of a scheduling policy on every interval \(]0:t]\), \(t\in\N\). Consequently, in the absence of communication uncertainties, our techniques provide improved numerical tractability compared to average dwell time based techniques.
	}
        \end{rem}
         %\todo[inline]{Reviewer 1, comment 3; Reviewer 7, comment 2: comparison with earlier works on stability of switched systems}
         \begin{rem}
        \label{rem:lit_diff}
        \rm{
            {In the recent past, multiple Lyapunov-like functions and graph-theoretic tools are widely used to construct stabilizing switching logics for switched systems, see e.g., \cite{abc,def,ghi}. A weighted directed graph is associated to a family of systems and the admissible transitions between them, and a switching logic is expressed as an infinite walk on this weighted directed graph. Infinite walks whose corresponding switching logics preserve stability, are constructed by employing negative weight cycles, see \cite[\S 3]{def}, \cite[\S 3]{abc}, \cite[\S 3]{ghi} for details. In our current paper, instead of studying GAS of \emph{a} switched system, we analyze ``simultaneous'' GAS of \(N\) switched systems each containing one asymptotically stable and one unstable subsystem. For that purpose, a stabilizing scheduling policy is designed by incorporating multiple switching logics, each of which is stabilizing. Not surprisingly, the design of \(T\)-contractive cycles transcends beyond identifying negative weight cycles on a weighted directed graph: it involves selection of \(T\)-factors that preserve GAS of all \(N\) plants, where every \(T\)-factor adds to the negativity of \(\Xi_{i}(W)\) for \(M\) plants and to the positivity of \(\Xi_{i}(W)\) for the remaining \(N-M\) plants. In addition, so far in the literature, negative weight cycles for stability of \emph{a} switched system are designed under the assumption that the Lyapunov-like functions \(V_{p}\), \(p\in\{\is,\iu\}\) and the corresponding scalars \(\lambda_{p}\), \(p\in\{\is,\iu\}\), \(i\in\{1,2,\ldots,N\}\) are ``given'', see e.g., \cite[Remark 9]{def}, \cite[Remark 9]{abc}, \cite[\S 2.2]{ghi} for discussions. In contrast, in the present work we deal with the harder problem of identifying \(T\)-contractive cycles on \(\G\), and design multiple Lyapunov-like functions \(V_{p}\) and the corresponding scalars \(\lambda_{p}\), \(p\in\{\is,\iu\}\), \(i=1,2,\ldots,N\) such that these cycles exist. }
       %\end{itemize}
        }
        \end{rem}
        %\todo[inline]{Reviewer 8: Comparison with suggested references, highlighting optimality.}
        \begin{rem}
        \label{rem:optimality}
        \rm{
        		{Optimal scheduling policies for remote state estimation in sensor networks are studied recently in \cite{Han'17, Li'18, Leong'17}. In the context of our results, one can utilize properties of \(T\)-contractive cycles on \(\G\) to achieve optimal stability margin for a scheduling policy. Notice that the choice of \(T\)-factors for a \(T\)-contractive cycle on the underlying weighted directed graph of the NCS under consideration, is not unique. Additionally, the choice of a \(T\)-contractive cycle itself is not unique. It is clear that employing ``any'' \(T\)-contractive cycle \(W\) on \(\G\) is sufficient to construct a stabilizing periodic scheduling policy as far as GAS of each plant \(i\) in \eqref{e:sys_dyn} is concerned. Fix \(i\in\{1,2,\ldots,N\}\). Any \(T\)-contractive cycle yields \(\Xi_{i}(W) = -\varepsilon_{i}\) for some \(\varepsilon_{i} > 0\). We observe that as \(\varepsilon_{i}\) increases, the rate of convergence of \(\norm{x_{i}(t)}\) improves, see also experimental results in \S\ref{s:num_ex}.}%, employing a cycle \(W_{opt}\) that yields a maximally negative \(\Xi_{i}(W)\), \(i=1,2,\ldots,N\) over all \(T\)-contractive cycles \(W = v_{0},(v_{0},v_{1}),v_{1},\ldots,v_{n-1},(v_{n-1},v_{0}),v_{0}\), \(n\geq 2\) and all choice of \(T\)-factors \(T_{v_{j}}\), \(j=0,1,\ldots,n-1\) offers a better convergence rate for \(\norm{x_{i}(t)}\), \(i\in 1,2,\ldots,N\).}
        }
        \end{rem}
        \begin{rem}
        \label{rem:non-unique_cycle}
        \rm{
            {The non-uniqueness of \(T\)-factors and \(T\)-contractive cycles described in Remark \ref{rem:optimality}} can be exploited to extend our results to the setting of a static scheduling policy  with a non-periodic structure. Indeed, suppose that \(W_{1}\) and \(W_{2}\) are two distinct (different in terms of either \(T\)-factors or vertices) \(T\)-contractive cycles on \(\G\). Then a scheduling policy of non-periodic structure can be generated by concatenating \(W_{1}\) and \(W_{2}\), e.g., \(W_{1}W_{2}W_{1}W_{2}W_{2}W_{1}W_{2}W_{2}W_{2}\ldots\). Such a scheduling policy is static because the allocation sequences of the shared communication channel are computed offline, but the sequences are applied in a non-periodic manner.
        }
        \end{rem}   
    \section{Numerical experiments}
\label{s:num_ex}
%\subsection{Experiment 1}
%\label{num_ex:exp1}
%
\subsection{Experiment 1}
\label{num_ex:exp1}
\subsubsection{The NCS}
	Consider an NCS with \(N=5\) discrete-time linear plants and a shared communication channel of limited capacity. The matrices \(A_{i}\in\R^{2\times 2}\), \(B_{i}\in\R^{2\times 1}\) and \(K_{i}\in\R^{1\times 2}\), \(i = 1,2,3,4,5\) are chosen as follows, numerical values are given in Table \ref{tab:matrix_table}.
	\begin{itemize}[label = \(\circ\), leftmargin = *,noitemsep,nolistsep]
		\item Elements of \(A_{i}\) are selected from the interval \([-2,2]\) uniformly at random.
		\item Elements of \(B_{i}\) are selected by picking values from the {\(\{0,1\}\)}.%interval \([0,1]\) uniformly at random, if the value is smaller than \(0.5\), we set it to \(0\), otherwise it is set to \(1\).
		\item It is ensured that the pair \((A_{i},B_{i})\) is controllable; \(K_{i}\) is the discrete-time linear quadratic regulator for \((A_{i},B_{i})\) with {\(Q_{i} = Q = 5I_{2\times 2}\) and \(R_{i} = R = 1\) }.
	\end{itemize}
	Suppose that \(M = 2\) plants are allowed to access the communication channel at every instant of time. %We will demonstrate GAS of all plants in the NCS under a periodic scheduling policy proposed in \S\ref{s:stab_sched}.
	\begin{table*}[htbp]
	\centering
	\begin{tabular}{|c | c | c|c|c|c|}
		\hline
		\(i\) & \(A_{i}\) & \(B_{i}\) & \(K_{i}\) & \(\abs{\lambda(A_{i})}\) & \(\abs{\lambda(A_{i}+B_{i}K_{i})}\)\\
		\hline
		\(1\) & \(\pmat{1.0310 & 0.9725\\-0.4311 &  0.6219}\) & \(\pmat{1\\0}\) & \(\pmat{-0.9869 & -0.7541}\) & \(1.0298,1.0298\) & \(0.3487,0.3487\)\\
		\hline
		\(2\) & \(\pmat{0.8375 & 1.0187\\-0.8959 &  0.7188}\) & \(\pmat{0\\1}\) & \(\pmat{0.4978 & -1.0887}\) & \(1.2307,1.2307\) & \(0.3095,0.3095\)\\
		\hline
		\(3\) & \(\pmat{1.2571 & -1.0259\\1.7171 &  -0.6001}\) & \(\pmat{1\\0}\) & \(\pmat{-0.7247 & 0.8152}\) & \(1.0036,1.0036\) & \(0.2056,0.2056\)\\
		\hline
		\(4\) & \(\pmat{0.7569 & 0.9926\\-0.1978 &  -1.6647}\) & \(\pmat{1\\1}\) & \(\pmat{-0.0933 & 0.8329}\) & \(0.6729,1.5807\) & \(0.0826,0.2508\)\\
		\hline
		\(5\) & \(\pmat{0.5294 & -1.6098\\-0.8860 &  0.1875}\) & \(\pmat{0\\1}\) & \(\pmat{0.9852 & -0.6016}\) & \(1.5649,0.8480\) & \(0.3085,0.1932\)\\
		\hline
	\end{tabular}
    \vspace*{0.2cm}
	\caption{Description of individual plants in the NCS}\label{tab:matrix_table}
	\end{table*}
    \begin{figure}[htbp]
	\begin{center}
		\includegraphics[scale = 0.4]{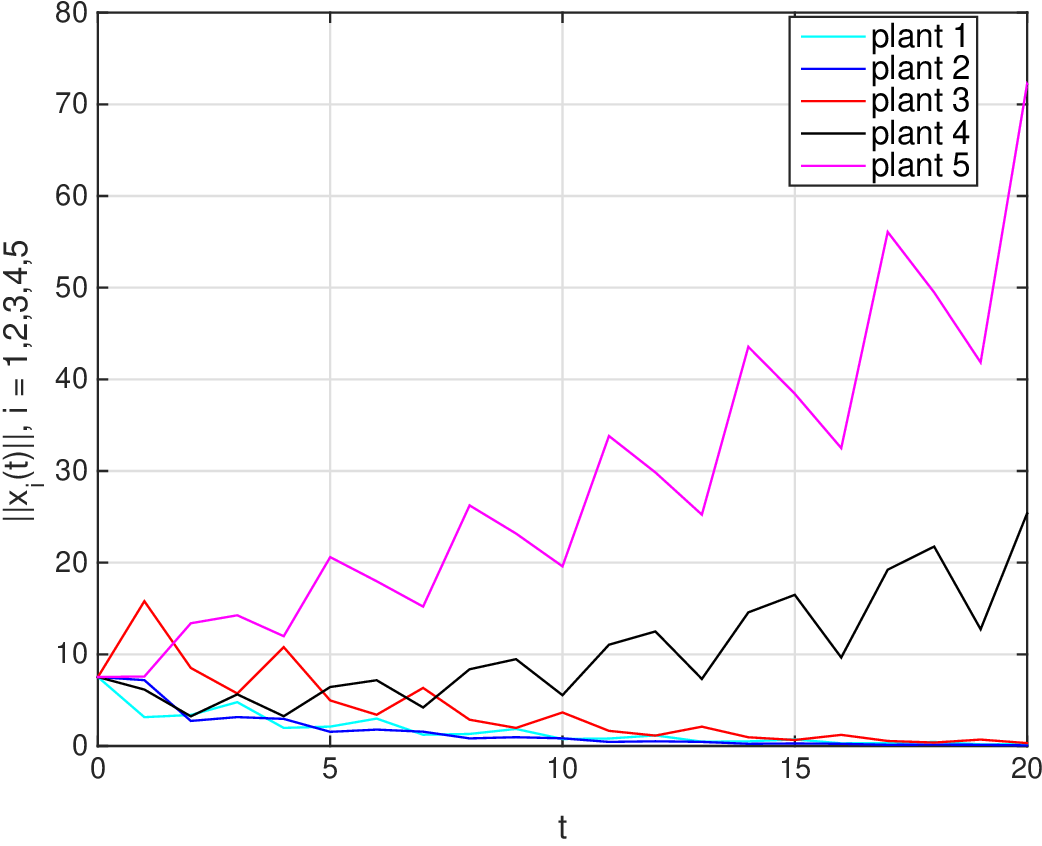}
		\caption{Not all plants are GAS under round-robin scheduling}\label{fig:counter_plot}
	\end{center}
	\end{figure}
\subsubsection{Non-triviality}
    We note that designing a stabilizing scheduling policy in the above setting is not a trivial problem. Indeed, consider a round-robin scheduling policy plants \(1\) and \(2\) followed by plants 2 and 3 followed by plants 4 and 5 accessing the channel, each combination being active for \(1\) unit of time. In Figure \ref{fig:counter_plot} we demonstrate that the plants \(4\) and \(5\) are unstable under this scheduling policy, and consequently, a careful design of \(\gamma\) is essential.
\subsubsection{The underlying weighted directed graph}
	We construct the underlying directed graph \(\G(\V,\E)\) of the NCS under consideration. For the given setting we have \({N\choose M} = 10\). \(\G\) consists of:
			\begin{itemize}[label = \(\circ\), leftmargin = *]
				\item \(\V = \{\overline{v}_{1}, \overline{v}_{2},\ldots,\overline{v}_{10}\}\) with\\
				\(L(\overline{v}_{1}) = \{1_{s},2_{s},3_{u},4_{u},5_{u}\}\), \(L(\overline{v}_{2}) = \{1_{s},2_{u},3_{s},4_{u},5_{u}\}\),
				\(L(\overline{v}_{3}) = \{1_{s},2_{u},3_{u},4_{s},5_{u}\}\), \(L(\overline{v}_{4}) = \{1_{s},2_{u},3_{u},4_{u},5_{s}\}\),\\
				\(L(\overline{v}_{5}) = \{1_{u},2_{s},3_{s},4_{u},5_{u}\}\), \(L(\overline{v}_{6}) = \{1_{u},2_{s},3_{u},4_{s},5_{u}\}\),
				\(L(\overline{v}_{7}) = \{1_{u},2_{s},3_{u},4_{u},5_{s}\}\), \(L(\overline{v}_{8}) = \{1_{u},2_{u},3_{s},4_{s},5_{u}\}\),\\
				\(L(\overline{v}_{9}) = \{1_{u},2_{u},3_{s},4_{u},5_{s}\}\), \(L(\overline{v}_{10}) = \{1_{u},2_{u},3_{u},4_{s},5_{s}\}\), and
				\item \(\E = \{(\overline{v}_{p},\overline{v}_{q}),\:\:p,q = 1,2,\ldots,10,\:\:p\neq q\}\).
			\end{itemize}
%	A pictorial representation of \(\G\) is shown below.
%    \begin{center}
%            \scalebox{1}{
%        \begin{tikzpicture}[every path/.style={>=latex},base node/.style={draw,circle}]
%            \node[base node]            (a) at (-1.5,0)  { \(v_1\) };
%            \node[base node]            (b) at (1.5,0)  { \(v_2\) };
%            \node[base node]            (c) at (0,-2) { \(v_3\) };
%
%	        \draw[->] (a) edge (b);
%            \draw[->] (b) edge[bend right] (a);
%	        \draw[->] (a) edge (c);
%            \draw[->] (c) edge[bend left] (a);
%            \draw[->] (b) edge (c);
%            \draw[->] (c) edge[bend right] (b);
%        \end{tikzpicture}}
%    \end{center}
\subsubsection{A \(T\)-contractive cycle}
	Fix a candidate contractive cycle \(W = v_{0},(v_{0},v_{1}),v_{1},(v_{1},v_{2}),v_{2},(v_{2},v_{0}),v_{0}\) on \(\G\), where \(v_{0} = \overline{v}_{5}\), \(v_{1} = \overline{v}_{4}\), \(v_{2} = \overline{v}_{10}\). We apply Algorithm \ref{algo:scalar_design_algo} with \(h_{s} = 0.0001\) and \(h_{u} = 0.1\), and obtain that \(W\) is \(T\)-contractive with \(T\)-factors: \(T_{v_{0}} = 4\), \(T_{v_{1}} = 3\), \(T_{v_{2}} = 5\). Indeed, \(\Xi_{1}(W) = -2.7629\), \(\Xi_{2}(W) = -8.0877\), \(\Xi_{3}(W) = -7.9572\), \(\Xi_{4}(W) = -0.2626\), \(\Xi_{5}(W) = -5.8414\). The corresponding values of the scalars \(\lambda_{p}\), \(p\in\{\is,\iu\}\) and \(\mu_{pq}\), \(p,q\in\{\is,\iu\}\), \(i=1,2,\ldots,N\) are given in Table \ref{tab:scalar_table}.
	\begin{table}[htbp]
	\centering
	\begin{tabular}{|c | c | c|c|c|c|}
		\hline
		\(i\) & \(\lambda_{\is}\) & \(\lambda_{\iu}\) & \(\mu_{\is\iu}\) & \(\mu_{\iu\is}\)\\
		\hline
		\(1\) & \(0.1360\) & \(1.2346\) & \(2.8452\) & \(1.3232\)\\
		\hline
		\(2\) & \(0.0720\) & \(1.2346\) & \(1.5681\) & \(1.3509\)\\
		\hline
		\(3\) & \(0.0715\) & \(1.2346\) & \(1.9025\) & \(1.3046\)\\
		\hline
		\(4\) & \(0.1757\) & \(2.7778\) & \(3.0854\) & \(1.1665\)\\
		\hline
		\(5\) & \(0.2430\) & \(2.7778\) & \(3.4664\) & \(1.1576\)\\
		\hline
	\end{tabular}
    \vspace*{0.2cm}
	\caption{Description of scalars admitting a solution to feasibility problem \eqref{e:feasprob2}}\label{tab:scalar_table}
	\end{table}
\subsubsection{The scheduling policy}
	A scheduling policy \(\gamma\) is obtained from Algorithm \ref{algo:sched_algo}. \(\gamma\) is constructed by employing \(W\), and it is periodic with period \(T_{v_{0}} + T_{v_{1}} + T_{v_{2}} = 12\) units of time. In Figure \ref{fig:gam_plot} we illustrate \(\gamma\) until time \(t = 60\).
    %\todo[inline]{Reviewer 1, comment 10: Figure 2 corrected.}
	\begin{figure}[htbp]
	\begin{center}
		\includegraphics[scale=0.4]{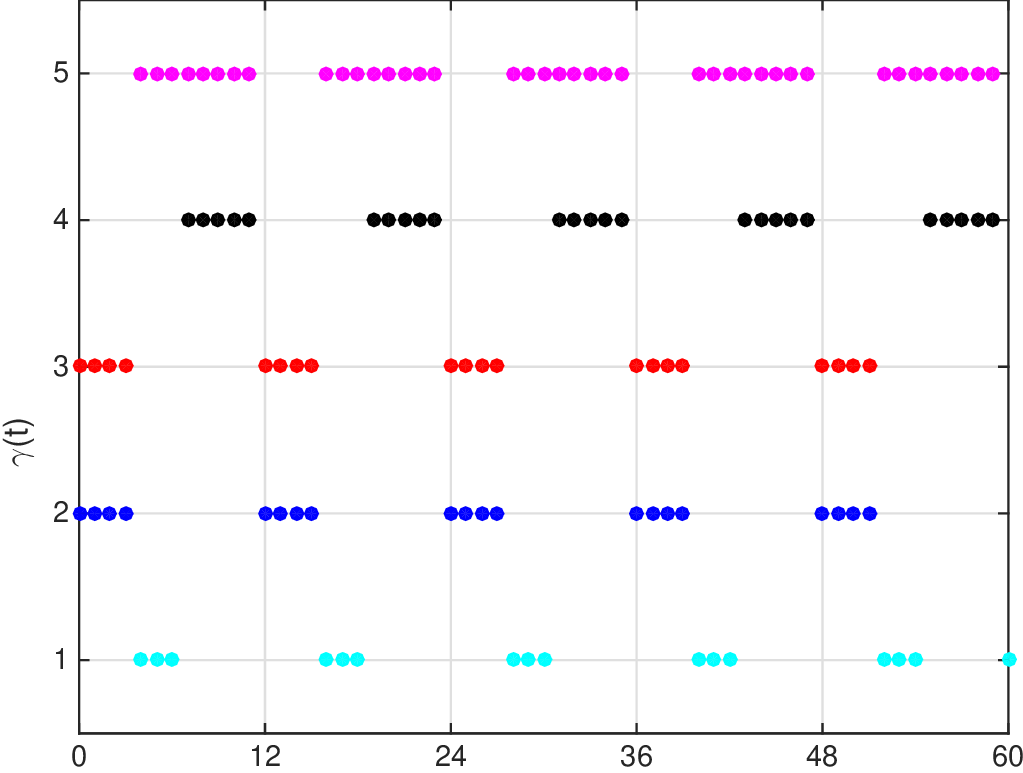}
		\caption{Scheduling policy \(\gamma\) obtained from Algorithm \ref{algo:sched_algo}}\label{fig:gam_plot}
	\end{center}
	\end{figure}
\subsubsection{GAS of NCS}
    We choose 100 different initial conditions from the interval \([-10,10]^{2}\) uniformly at random, and plot \((\norm{x_{i}(t)})_{t\in\N_{0}}\) under the scheduling policy \(\gamma\), \(i = 1,2,3,4,5\). Figure \ref{fig:xt_plot} contains plots for \(\norm{x_{i}(t)}\), \(i = 1,2,3,4,5\) until time \(t = 60\). It is observed that the individual plants of the NCS under consideration are GAS under our scheduling policy.
	\begin{figure*}[htbp]
	\begin{center}
		\includegraphics[height = 4cm, width = 5cm]{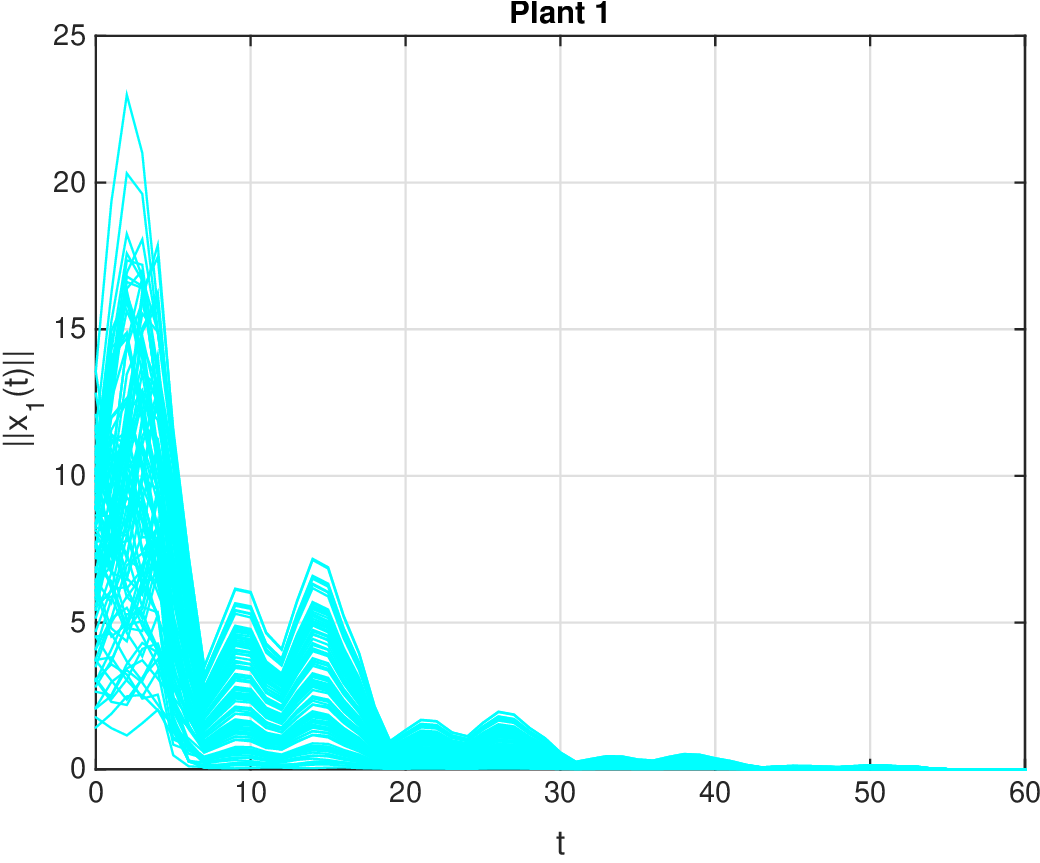}\hspace*{0.5cm}\includegraphics[height = 4cm, width = 5cm]{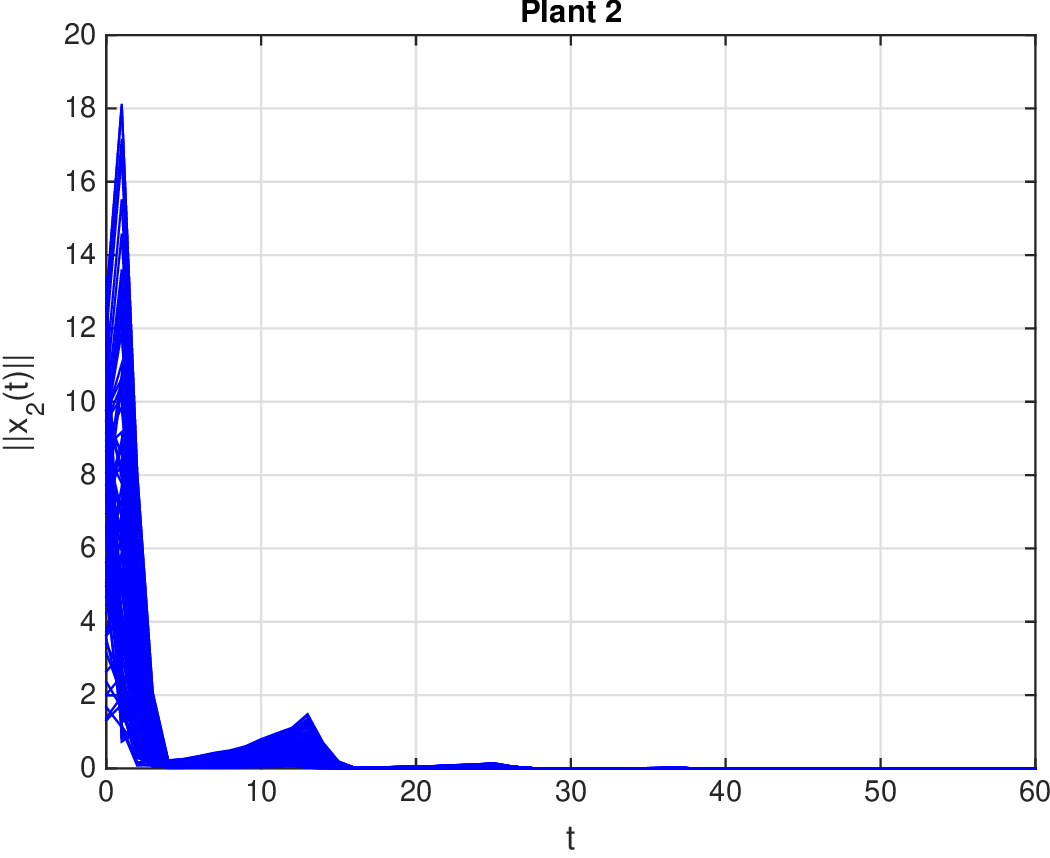}
		\includegraphics[height = 4cm, width = 5cm]{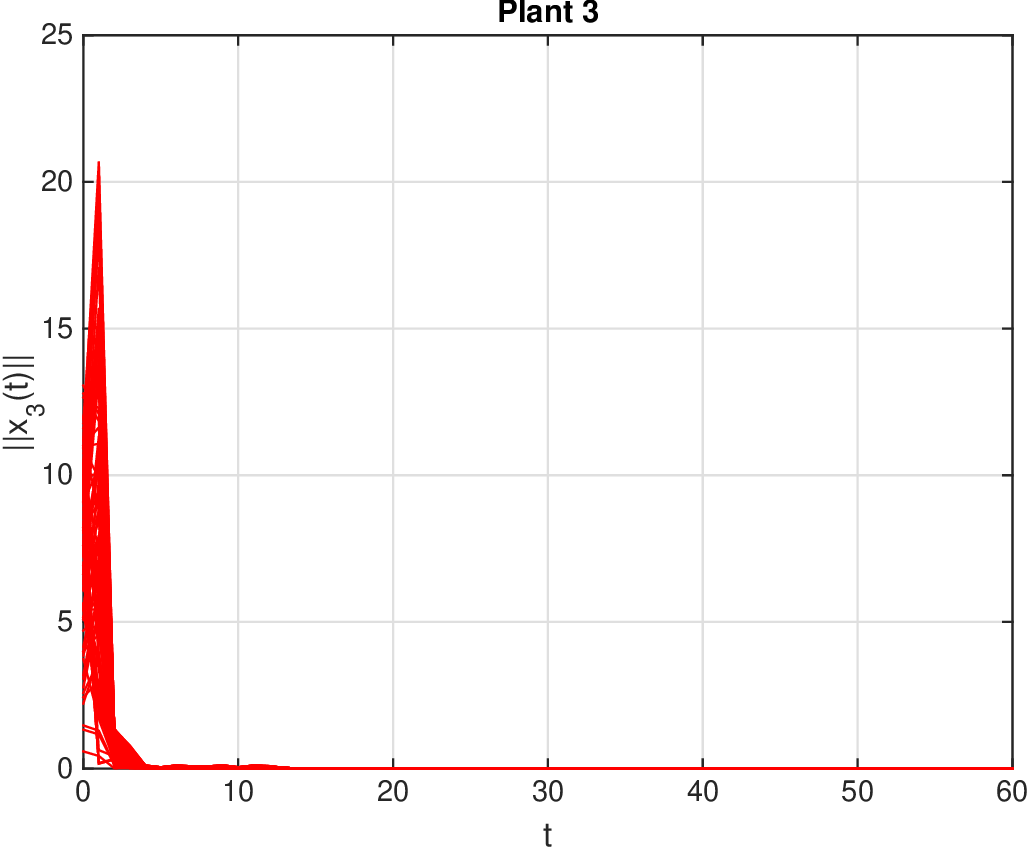}\\\
		\includegraphics[height = 4cm, width = 5cm]{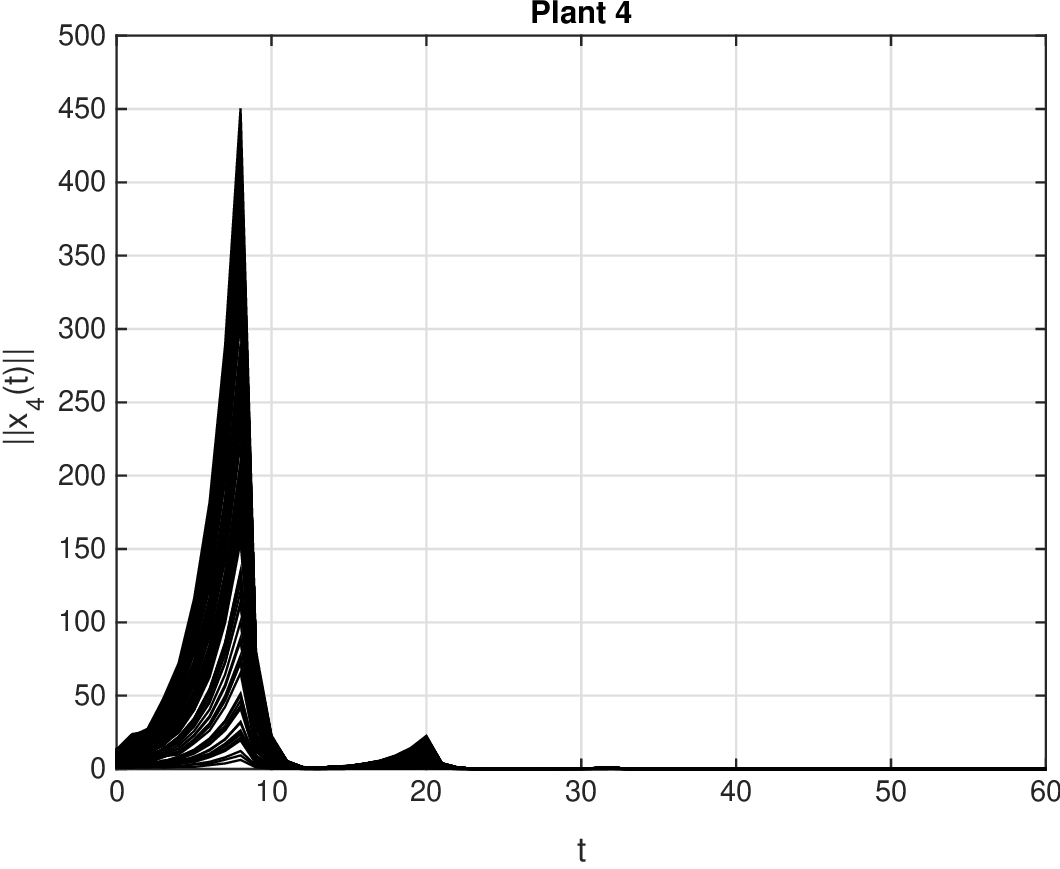}
		\includegraphics[height = 4cm, width = 5cm]{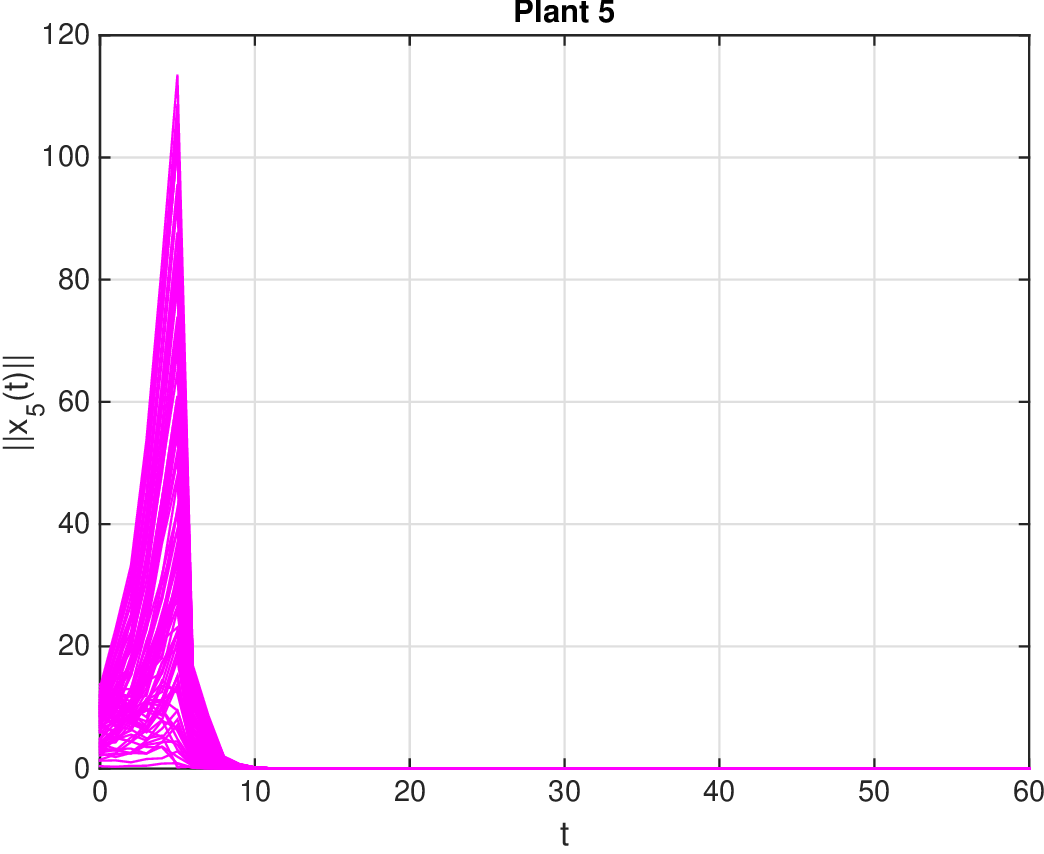}
		\caption{Plot for \(\norm{x_{i}(t)}\) versus \(t\) for each plant \(i=1,2,3,4,5\)}\label{fig:xt_plot}
	\end{center}
	\end{figure*}
 %\todo[inline]{Reviewer 1, comment 9: performance comparison with respect to different \(T\)-contractive cycles}
{{
\subsubsection{Comparison}
    We choose three distinct \(T\)-contractive cycles \(W_{j} = v_{0}^{(j)}, (v_{0}^{(j)},v_{1}^{(j)}),v_{1}^{(j)},(v_{1}^{(j)},v_{2}^{(j)}),v_{2}^{(j)},(v_{2}^{(j)},v_{0}^{(j)}),v_{0}^{(j)}\) on \(\G\). The description of the cycles and the corresponding values of \(\Xi_{i}(W_{j})\), \(j = 1,2,3\), \(i=1,2,3,4,5\) are given in Table \ref{tab:compa_cycle}.
    \begin{table*}[htbp]
	\centering
	\begin{tabular}{|c|c|c|c|c|c|c|c|c|c|c|c|}
		\hline
		\(j\) & \(v_{0}\) & \(v_{1}\) & \(v_{2}\) & \(T_{v_{0}}\) & \(T_{v_{1}}\) & \(T_{v_{2}}\) & \(\Xi_{1}(W_{j})\) & \(\Xi_{2}(W_{j})\) & \(\Xi_{3}(W_{j})\) & \(\Xi_{4}(W_{j})\) & \(\Xi_{5}(W_{j})\)\\
		\hline
		\(1\) & \(\overline{v}_{5}\) & \(\overline{v}_{3}\) & \(\overline{v}_{9}\) & \(2\) & \(7\) & \(8\) & \(-10.5325\) & \(-1.3503\) & \(-23.9963\) & \(-0.67556\) & \(-0.73315\)\\
		\hline
		\(2\) & \(\overline{v}_{2}\) & \(\overline{v}_{6}\) & \(\overline{v}_{7}\) & \(3\) & \(8\) & \(9\) & \(-1.0769\) & \(-43.3456\) & \(-3.4224\) & \(-0.37122\) & \(-0.10453\)\\
		\hline
		\(3\) & \(\overline{v}_{8}\) & \(\overline{v}_{9}\) & \(\overline{v}_{1}\) & \(8\) & \(9\) & \(3\) & \(-1.0769\) & \(-3.5599\) & \(-43.3057\) & \(-0.37122\) & \(-0.10453\)\\
		\hline
	\end{tabular}
    \vspace*{0.2cm}
	\caption{Description of different \(T\)-contractive cycles on \(\G\)}\label{tab:compa_cycle}
	\end{table*}
    We now illustrate that with smaller values of \(\Xi_{i}(W_{j})\), the rate of convergence of \(\norm{x_{i}(t)}\) to \(0\) becomes faster. For this purpose, we pick 10 different initial conditions \(x_{i}(0)\) from the interval \([-1,1]^{2}\) uniformly at random and simulate \((\norm{x_{i}(t)})_{t\in\N_{0}}\) for the cycles \(W_{j}\). Figure \ref{fig:conv_plot} contains plots for \(\displaystyle{\average_{x_{i}(0)}}({\norm{x_{i}(t)}})_{t\in\N_{0}}\) for the plant \(i=3\) corresponding to the cycles \(W_{j}\).
    	\begin{figure}[htbp]
		\begin{center}
			\includegraphics[scale = 0.4]{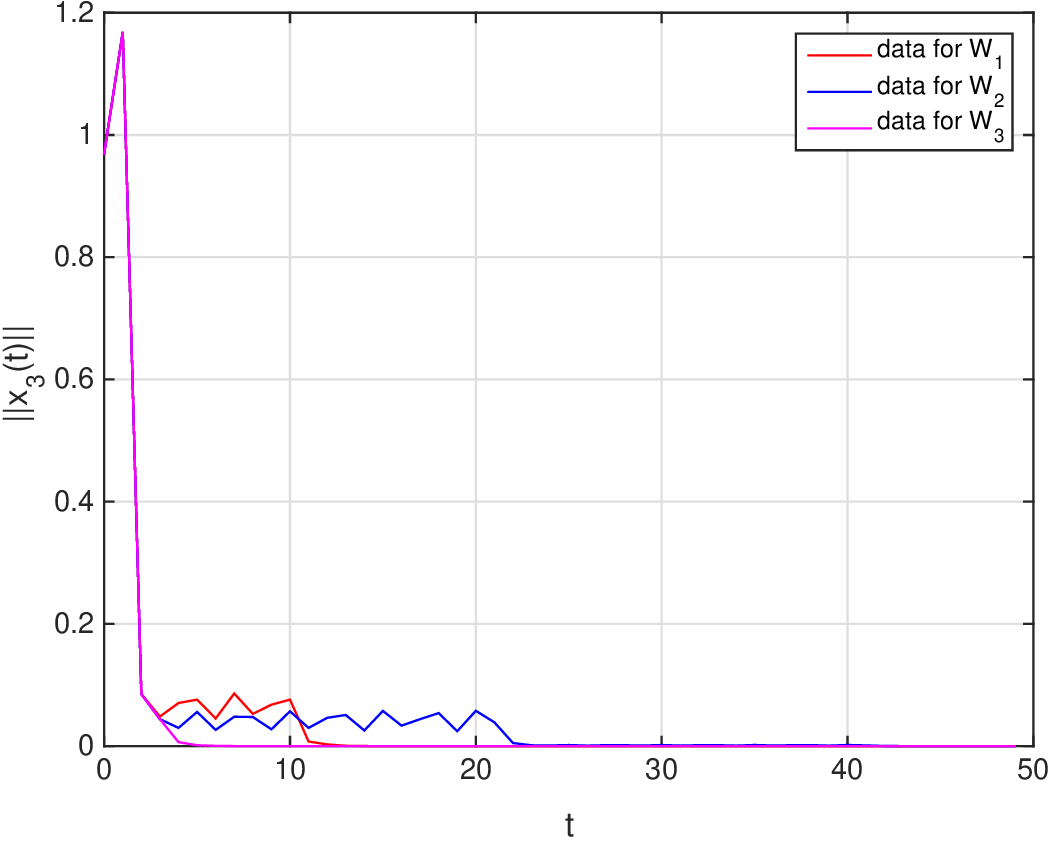}
		\end{center}
		\caption{Plot for \(\average_{x_{3}(0)}\norm{x_{3}(t)}\) versus \(t\) corresponding to cycles \(W_{j}\).}\label{fig:conv_plot}
	\end{figure}
\subsection{Experiment 2}
\label{num_ex:exp2}
   % \todo[inline]{Reviewer 1, comment 11: More experimental studies.}
    We now test the performance of our techniques in large scale settings. This involves a three-step procedure:
    \begin{itemize}[label = \(\circ\), leftmargin = *]
        \item First, we generate \(N\) non-Schur matrices \(A_{i}\in\R^{2\times 2}\) and vectors \(B_{i}\in\R^{2\times 1}\) with entries from the interval \([-2,2]\) and the set \(\{0,1\}\), respectively, chosen uniformly at random, and ensuring that each pair \((A_{i},B_{i})\), \(i=1,2,\ldots,N\), is controllable. The associated linear quadratic regulators \(K_{i}\) are computed with \(Q_{i} = Q = 5I_{2\times 2}\) and \(R_{i} = R = 1\).
        \item Second, the underlying directed graph \(\G\) of the NCS is considered, and a candidate contractive cycle \(W = v_{0},(v_{0},v_{1}),v_{1},\ldots\),\\\(v_{n-1},(v_{n-1},v_{0}),v_{0}\) on \(\G\) is chosen in the following manner:
            \begin{itemize}[label = \(\diamond\), leftmargin = *]
                \item Fix \(M = 10\). We pick \(10\) distinct numbers from the set \(\{1,2,\ldots,N\}\) uniformly at random, and repeat this process until each element of \(\{1,2,\ldots,N\}\) is picked at least once.
                \item For each set of \(10\) distinct numbers picked above, we choose a vertex \(v\) on \(\G\) that satisfies \(\ell_{v}(i) = \is\) for these numbers. \(W\) is constructed by concatenating these vertices neglecting repetitions.
            \end{itemize}
        \item Third, Algorithm \ref{algo:scalar_design_algo} is employed with the chosen \(A_{i}\), \(B_{i}\), \(K_{i}\) and \(W\) to design the scalars \(\lambda_{\is}\), \(\lambda_{\iu}\), \(\mu_{\is\iu}\), \(\mu_{\iu\is}\), \(i = 1,2,\ldots,N\) for \(T\)-contractivity of \(W\). The step-sizes are chosen as \(h_{s} = h_{u} = 0.001\).
    \end{itemize}
    We repeat the above procedure for large values of \(N\) and study the total computation times. The procedure is implemented on a MATLAB platform.\footnote{Version: R2015a, System specifications: Intel Core i7 processor, 8GB RAM, 64-bit OS.} In Table \ref{tab:graph_data} we list sizes of \(\G\) and lengths of candidate contractive cycles \(W = v_{0},(v_{0},v_{1}),v_{1},\ldots,v_{n-1}\),\\\((v_{n-1},v_{0}),v_{0}\) for various values of \(N\), and the computation times. In case of \(N = 700\), the first choice of candidate contractive cycle \(W\) did not turn out to be \(T\)-contractive from Algorithm \ref{algo:scalar_design_algo}. We then repeated the process of generating a candidate contractive cycle, and obtained a solution from Algorithm \ref{algo:scalar_design_algo} for the new choice of \(W\).
    \begin{table}[htbp]
	\centering
	\begin{tabular}{|c | c | c|c|}
		\hline
		\(N\) & \(\abs{\V}\) & \(n\) & Computation time (sec)\\
        \hline
        \(100\) & \(1.73\times 10^{13}\) & \(61\) & 4681.77\\
        \hline
        \(200\) & \(2.24\times 10^{16}\) & \(104\) & 7966.65\\
        \hline
        \(500\) & \(2.45\times 10^{20}\) & \(345\) & 35268.59\\
        \hline
        \(700\) & \(7.3\times 10^{21}\) & \(532\) & 76746.32\\
        \hline
        \(1000\) & \(2.63\times 10^{23}\) & \(822\) & 70453.28\\
		\hline
	\end{tabular}
    \vspace*{0.2cm}
	\caption{Graph and cycle data}\label{tab:graph_data}
	\end{table}
    %\begin{figure}
%    \centering
%        \includegraphics[scale = 0.45]{comp_time-eps-converted-to1}
%        \caption{Plot for total computation time versus \(N\)}\label{fig:comp_time}
%    \end{figure}
}} 
    \section{Concluding remarks}
\label{s:concln}
	In this paper we presented a stabilizing scheduling policy for NCSs under medium access constraints. A switched system representation is associated to the individual plants, and a weighted directed graph is associated to the NCS. Our scheduling policy is designed by employing a \(T\)-contractive cycle on the underlying weighted directed graph of the NCS. We also address algorithmic construction of \(T\)-contractive cycles. {Since our algorithm for designing \(T\)-contractive cycles does not conclude about their non-existence, an important question is regarding the design of such cycles when our algorithm does not yield a solution for all choices of candidate contractive cycle on a weighted directed graph.} Also, a natural extension of our work is to accommodate network induced uncertainties such as access delays, packet dropouts, etc. in the feedback control loop. These aspects are currently under investigation and will be reported elsewhere. 
    \section{Proofs of results}
\label{s:all_proofs}
%\section{Proof of Theorem \ref{t:mainres1}}
%\label{s:all_proofs}
	\begin{proof}[Proof of Theorem \ref{t:mainres1}]
		Consider the {NCS} described in \S\ref{s:prob_stat} and its underlying directed graph \(\G(\V,\E)\). Let \(W = v_{0},(v_{0},v_{1}),v_{1},\ldots\),\\\(v_{n-1},(v_{n-1},v_{0}),v_{0}\) be a \(T\)-contractive cycle on \(\G\). Consider a scheduling policy \(\gamma\) obtained from Algorithm \ref{algo:sched_algo} constructed by employing \(W\). We will show that each plant in \eqref{e:sys_dyn} is GAS under \(\gamma\).
		
		Fix an arbitrary plant \(i\in\{1,2,\ldots,N\}\). In view of the switched systems representation of plant \(i\) in \eqref{e:i-plant_sw-sys}, it suffices to show that the switching logic \(\sigma_{i}\) corresponding to \(\gamma\), ensures GAS of plant \(i\).
		
		Fix a time \(t\in\N\). Recall that \(0 =:\tau_{0}<\tau_{1}<\cdots\) are the points in time at which \(\gamma\) changes values. Let \(\Nsig\) be the total number of times \(\gamma\) {has} changed its values on \(]0:t]\). In view of \eqref{e:key_ineq1}, we have
		\begin{align}
		\label{e:pf1step1}
			V_{\si(t)}(\xp(t))\leq\lambda_{\si(\tau_{\Nsig})}^{t-\tau_{\Nsig}} V_{\si(\tau_{\Nsig})}(\xp(\tau_{\Nsig})).
		\end{align}
		A straightforward iteration of \eqref{e:pf1step1} using \eqref{e:key_ineq1} and \eqref{e:key_ineq2} gives
		\begin{align}
		\label{e:pf1step2}
			V_{\si(t)}(\xp(t))\leq\Biggl(\prod_{\substack{{j=0}\\{\tau_{\Nsig+1}:=t}}}^{\Nsig}\lambda_{\si(\tau_{j})}^{\tau_{j+1}-\tau_{j}}\cdot\prod_{j=0}^{\Nsig-1}\mu_{\si(\tau_{j})\si(\tau_{j+1})}\Biggr)V_{\si(0)}(\xp(0)).
		\end{align}
		The first term on the right-hand side of the above inequality can be rewritten as
		\begin{align*}
			\exp\Biggl(\ln\Biggl(\prod_{\substack{{j=0}\\{\tau_{\Nsig+1}:=t}}}^{\Nsig}\lambda_{\si(\tau_{j})}^{\tau_{j+1}-\tau_{j}}\Biggr) + \ln\Biggl(\prod_{j=0}^{\Nsig-1}\mu_{\si(\tau_{j})\si(\tau_{j+1})}\Biggr)\Biggr).
		\end{align*}
		Now,
		\begin{align}
		\label{e:pf1_step3_ia}			\ln\Biggl(\prod_{\substack{{j=0}\\{\tau_{\Nsig+1}:=t}}}^{\Nsig}\lambda_{\si(\tau_{j})}^{\tau_{j+1}-\tau_{j}}\Biggr) = \sum_{\substack{{j=0}\\{\tau_{\Nsig+1}:= t}}}^{\Nsig}(\tau_{j+1}-\tau_{j})\ln\lambda_{\si(\tau_{j})}
			=\sum_{\substack{{j=0}\\\tau_{\Nsig+1}:=t}}^{\Nsig}\biggl(\sum_{p\in\{\is,\iu\}}\mathrm{1}(\si(\tau_{j}) = p)(\tau_{j+1}-\tau_{j})\ln\lambda_{p}\biggr).
		\end{align}
		Let \(D_{s}(s,t)\) and \(D_{u}(s,t)\) denote the total durations {number of time-steps} of activation of the stable and unstable modes of \(i\) on \(]s:t]\), respectively. Recall that \(0 < \ls < 1\) and \(\lu\geq 1\). Consequently, \(\ln\ls < 0\) and \(\ln\lu \geq 0\). Thus, the right-hand side of \eqref{e:pf1_step3_ia} is equal to
		\begin{align}
		\label{e:pf1step3}
			-\abs{\ln\ls}D_{s}(0,t) + \abs{\ln\lu}D_{u}(0,t).
		\end{align}
		Let \(N_{pq}(s,t)\) denote the total number of transitions from subsystem (mode) \(p\) to subsystem (mode) \(q\), \(p,q\in\{\is,\iu\}\) on \(]s:t]\). We have
		\begin{align}
		\label{e:pf1step4}
			\ln\Biggl(\prod_{j=0}^{\Nsig-1}\mu_{\si(\tau_{j})\si(\tau_{j+1})}\Biggr) = \sum_{j = 0}^{\Nsig-1}\ln\mu_{\si(\tau_{j})\si(\tau_{j+1})}
			&= \sum_{p\in\{\is,\iu\}}\sum_{j=0}^{\Nsig-1}\sum_{\substack{{p\to q:}\\{q\in\{\is,\iu\},}\\{\si(\tau_{j}) = p,}\\{\si(\tau_{j+1}) = q}}}\ln\mu_{pq}\nonumber\\
			&= \ln\mu_{\is\iu}N_{\is\iu}(0,t) + \ln\mu_{\iu\is}N_{\iu\is}(0,t),
		\end{align}
        since \(\mu_{\is\is} = \mu_{\iu\iu} = 1\).
		Substituting \eqref{e:pf1step3} and \eqref{e:pf1step4} in \eqref{e:pf1step2}, we obtain
		\begin{align}
		\label{e:pf1step5}
			V_{\si(t)}(\xp(t))\leq\psi_{i}(t)V_{\si(0)}(\xp(0)),
		\end{align}
		where
		\begin{align}
		\label{e:psi_i_defn}
			\N\ni t\mapsto \psi_{i}(t) &:= \exp\Biggl(-\abs{\ln\ls}D_{s}(0,t)+\abs{\ln\lu}D_{u}(0,t)
+\ln\mu_{\is\iu}N_{\is\iu}(0,t)+\ln\mu_{\iu\is}N_{\iu\is}(0,t)\Biggr).
		\end{align}
		From the definition of \(V_{p}\), \(p\in\{\is,\iu\}\) in \eqref{e:Lyap-func_defn} and properties of positive definite matrices \cite[Lemma 8.4.3]{Bernstein}, it follows that
		\begin{align}
			\norm{\xp(t)}\leq c\psi_{i}(t)\norm{\xp(0)}\:\:\text{for all}\:\:t\in\N_{0},
		\end{align}
		where \(\displaystyle{c = \sqrt{\frac{\displaystyle{\max_{p\in\{\is,\iu\}}\lambda_{\max}(P_{p})}}{\displaystyle{\min_{p\in\{\is,\iu\}}\lambda_{\min}(P_{p})}}}}\), where for a matrix \(A\in\R^{d\times d}\), \(\lambda_{\min}(A)\) denotes the minimum eigenvalue of \(A\).
		By Definition \ref{d:gas}, {to establish} GAS of \eqref{e:i-plant_sw-sys}, we need to show that \(c\norm{\xp(0)}\psi_{i}(t)\) can be bounded above by a class \(\KL\) function. Towards this end, we already see that \(c\norm{\xp(0)}\) is a class \(\K_{\infty}\) function. Therefore, it remains to show that \(\psi_{i}(t)\) is bounded above by a function in class \(\cL\).
		
		Recall that \(\gamma\) is constructed by employing a \(T\)-contractive cycle \(W = v_{0},(v_{0},v_{1}),v_{1},\ldots,v_{n-1},(v_{n-1},v_{0}),v_{0}\) on \(\G\), and \(T_{v_{j}}\), \(j = 0,1,\ldots,n-1\) are the \(T\)-factors associated to vertices \(v_{j}\), \(j = 0,1,\ldots,n-1\). Let \(\displaystyle{T_{W} := \sum_{j=0}^{n-1}T_{v_{j}}}\), \(t\geq mT_{W}\), \(m\in\N_{0}\), and \(\Xi_{i}(W) = -\varepsilon_{i}\), \(\varepsilon_{i} > 0\), where \(\Xi_{i}(W)\) is as defined in \eqref{e:contra_defn}.
		By construction of \(\gamma\), we have
		\begin{align}
		\label{e:pf1step7}
			\psi_{i}(t) &= \exp\Biggl(-\abs{\ln\ls}D_{s}(0,t)+\abs{\ln\lu}D_{u}(0,t)+\ln\mu_{\is\iu}N_{\is\iu}(0,t)+\ln\mu_{\iu\is}N_{\iu\is}(0,t)\Biggr)\nonumber\\
			&=-\abs{\ln\ls}D_{s}(0,mT_{W})-\abs{\ln\ls}D_{s}(mT_{W},t)+\abs{\ln\lu}D_{u}(0,mT_{W})+\abs{\ln\lu}D_{u}(mT_{W},t)\nonumber\\
			&\:\:\:\:\:+\ln\mu_{\is\iu}N_{\is\iu}(0,mT_{W})+\ln\mu_{\is\iu}N_{\is\iu}(mT_{W},t)+\ln\mu_{\iu\is}N_{\iu\is}(0,mT_{W})+\ln\mu_{\iu\is}N_{\iu\is}(mT_{W},t).
		\end{align}
		Notice that
		\begin{align*}
			&-\abs{\ln\ls}D_{s}(0,mT_{W}) + \abs{\ln\lu}D_{u}(0,mT_{W})+ \ln\mu_{\is\iu}N_{\is\iu}(0,mT_{W}) + \ln\mu_{\iu\is}N_{\iu\is}(0,mT_{W})\\
			&=-\abs{\ln\ls}m\sum_{\substack{{j:\l_{v_{j}}(i) = \is}\\{j = 0,1,\ldots,n-1}}}T_{v_{j}} + \abs{\ln\lu}m\sum_{\substack{{j:\l_{v_{j}}(i) = \iu}\\{j = 0,1,\ldots,n-1}}}T_{v_{j}}+\ln\mu_{\is\iu}m\#(\is\to\iu)_{W}+\ln\mu_{\iu\is}m\#(\iu\to\is)_{W},
		\end{align*}
		where \(\#(p\to q)_{W}\) denotes the number of times a transition from a vertex \(v_{j}\) to a vertex \(v_{j+1}\) has occurred in \(W\) such that \(\l_{v_{j}}(i) = p\) and \(\l_{v_{j+1}}(i) = q\), \(p,q\in\{\is,\iu\}\), \(p\neq q\). The right-hand side of the above equality can be rewritten as
		\begin{align}
		\label{e:pf1step8a}
			m\Biggl(-\abs{\ln\ls}\sum_{\substack{{j:\l_{v_{j}}(i) = \is}\\{j = 0,1,\ldots,n-1}}}T_{v_{j}} &+ \abs{\ln\lu}\sum_{\substack{{j:\l_{v_{j}}(i) = \iu}\\{j = 0,1,\ldots,n-1}}}T_{v_{j}}+\ln\mu_{\is\iu}\#(\is\to\iu)_{W} + \ln\mu_{\iu\is}\#(\iu\to\is)_{W}\Biggr).%-\abs{\ln\ls}\sum_{\substack{{j:\l_{v_{j}}(i) = \is}\\{j = 0,1,\ldots,n-1}}T_{v_{j}
		\end{align}
		From the definition of weights associated to vertices and edges of \(\G\), we have that the above expression is equal to \(-m\varepsilon_{i}\).
		Also,
		\begin{align}
		\label{e:pf1step8b}
			&-\abs{\ln\ls}D_{s}(mT_{W},t) + \abs{\ln\lu}D_{u}(mT_{W},t)+\ln\mu_{\is\iu}N_{\is\iu}(mT_{W},t)+\ln\mu_{\iu\is}N_{\iu\is}(mT_{W},t)\nonumber\\
			&\leq\abs{\ln\lu}(t-mT_{W}) + mn(\ln\mu_{\is\iu}+\ln\mu_{\iu\is}) := a\:\text{(say)}.
		\end{align}
		From \eqref{e:pf1step8a} and \eqref{e:pf1step8b}, we obtain that the right-hand side of \eqref{e:pf1step7} is bounded above by \(\exp\bigl(-m\varepsilon_{i}+a\bigr)\).
		
		Let \(\varphi_{i}:[0,\:t]\to\R\) be a function connecting \((0,\exp(a)+T_{W})\), \((rT_{W},\exp(-(r-1)\varepsilon_{i}+a))\), \((t,\exp(-m\varepsilon_{i}+a))\), \(r = 1,2,\ldots,m\), with straight line segments. By construction, \(\varphi_{i}\) is an upper envelope of \(T\mapsto\psi_{i}(T)\) on \([0{,} t]\), is continuous, decreasing, and tends to \(0\) as \(t\to+\infty\). Hence, \(\varphi_{i}\in\cL\).
		
		Recall that \(i\in\{1,2,\ldots,N\}\) was selected arbitrarily. It follows that our assertion holds for all plants \(i\) in \eqref{e:sys_dyn}.
	\end{proof}

    %\todo[inline]{Reviewer 6, comment 3): Justification for using absolute values}
    \begin{rem}
    \label{rem:wt_choice2}
    \rm{
        {Our scheduling policy \(\gamma\) is constructed by employing \(T\)-contractive cycles on \(\G\). The definition of the functions \(\psi_{t}\), \(i=1,2,\ldots,N\) clarifies the association of natural logarithm with the scalars \(\lambda_{p}\), \(p\in\{\is,\iu\}\) and \(\mu_{pq}\), \(p,q\in\{\is,\iu\}\) in the vertex and edge weights of \(\G\), respectively, \(i=1,2,\ldots,N\). The use of absolute values with \(\ln\lambda_{p}\), \(p\in\{\is,\iu\}\), \(i=1,2,\ldots,N\) allows for an easy distinction between the positive and negative terms in \(\psi_{i}\), \(i=1,2,\ldots,N\). Alternatively, one may exclude the absolute values and keep track of which terms are negative in the subsequent analysis.
        }
    }
    \end{rem}
%======================================================
%======================================================
\begin{proof}[Proof of Proposition \ref{prop:mainres3}]
{
{
    Let \(M = 1\). Fix a cycle \(W = v_{0},(v_{0},v_{1}),v_{1},\ldots,v_{N-1},(v_{N-1},v_{0}),v_{0}\) on \(\G\) that satisfies \(v_{k}^{s}\cap v_{\l}^{s} = \emptyset\) for all \(k,\l = 0,1,\ldots,N-1\), \(l\neq\l\). Clearly, \(W\) is a candidate contractive cycle on \(\G\).

    Without loss of generality, let us assume that \(\l_{v_{i-1}}(i) = \is\), \(i=1,2,\ldots,N\). Suppose that \(T_{v_{j}} = \tilde{T}\), \(j=0,1,\ldots,N-1\). By construction of \(W\), the left-hand side of \eqref{e:main_ineq1} is
    \begin{align*}
        &-\abs{\ln\lambda_{\is}}T_{v_{i-1}} +  \abs{\ln\lambda_{\iu}}\Bigl(\sum_{\substack{j=0\\j\neq i-1}}^{N-1}T_{v_{j}}\Bigr) + \ln\mu_{\is\iu} + \ln\mu_{\iu\is}\\
        =& \Bigl(-\abs{\ln\lambda_{\is}} + (N-1)\abs{\ln\lambda_{\iu}}\Bigr)\tilde{T} + \ln\mu_{\is\iu} + \ln\mu_{\iu\is},\:i=1,2,\ldots,N.
    \end{align*}

    Since \eqref{e:suff_condn1} holds, it is possible to choose an integer \(\tilde{T} > 0\) such that the above expression is strictly less than \(0\).}}
\end{proof}
%=======================================================
%=======================================================
\begin{proof}[Proof of Proposition \ref{prop:mainres4}]
{{
    Let \(\displaystyle{M \geq {N}/{2}}\). Fix a cycle \(W = v_{0},(v_{0},v_{1}),v_{1},(v_{1},v_{0}),v_{0}\) on \(\G\) that satisfies \(v_{1}^{s}\supset\{1,2,\ldots,N\}\setminus V_{0}^{s}\). Let \(j_{1},j_{2},\ldots,j_{M}\) and \(k_{1},k_{2},\ldots,k_{N-M}\in\{1,2,\ldots,N\}\) be the elements for which \(\l_{v_{0}}(j_{p}) = j_{p_{s}}\), \(p=1,2,\ldots,M\) and \(\l_{v_{1}}(k_{q}) = k_{q_{s}}\), \(q = 1,2,\ldots,N-M\), respectively. We have \(\abs{\{j_{1},j_{2},\ldots,j_{M}\}}\geq\abs{k_{1},k_{2},\ldots,k_{N-M}}\). It is immediate that \(W\) is a candidate contractive cycle.

    Suppose that \(T_{v_{0}} = T_{v_{1}} = \tilde{T}\). By construction of \(W\), we have \(\overline{N}_{\is\iu}\), \(\overline{N}_{\iu\is}\in\{0,1\}\), \(i=1,2,\ldots,N\). The left-hand side of \eqref{e:main_ineq1} is bounded above by
    \[
        \Bigl(-\abs{\ln\lambda_{i_{s}}}+\abs{\ln\lambda_{\iu}}\Bigr)\tilde{T} + \ln\mu_{\is\iu} + \ln\mu_{\iu\is}.
    \]

    Since condition \eqref{e:suff_condn2} holds, there exists \(\tilde{T} > 0\) such that the above expression is strictly less than \(0\).}}
\end{proof}

\end{document}